\documentclass[a4paper,UKenglish,cleveref, autoref, thm-restate]{lipics-v2021}

\pdfoutput=1 %
\hideLIPIcs  %

\usepackage{amsmath}
\usepackage{amsthm}
\usepackage{adjustbox}
\usepackage[margin=0pt]{caption}
\usepackage{changepage}
\usepackage{setspace}
\usepackage{tikz}
\usepackage{xcolor}
\usepackage{paralist}
\usepackage{paralist}
\usepackage{cancel}
\usetikzlibrary{arrows.meta}
\usetikzlibrary{shapes,snakes}
\usepackage{bbding}
\usepackage{amstext}
\usepackage{tikz}
\usetikzlibrary{shapes,snakes}
\usepackage{amsmath}

\usepackage{array}
\usepackage{tabu}
\usepackage[T1]{fontenc}
\usepackage{array}
\usepackage{amssymb}
\usepackage{amsfonts}
\usepackage{xspace}

\usepackage{amsthm}
\usepackage{algorithm}
\usepackage[noend]{algorithmic}
\usepackage{microtype}
\usepackage[colorinlistoftodos,prependcaption,textsize=tiny]{todonotes}
\usepackage{hyperref}

\title{On Instance-Optimal Algorithms for a Generalization of \\Nuts and Bolts and Generalized Sorting}

\titlerunning{Instance-Optimal Algorithm and Sorting} %
\author{Mayank Goswami}{Queens College CUNY, Flushing, New York, USA}{mayank.goswami@qc.cuny.edu}{https://orcid.org/0000-0002-2111-3210}{Supported by NSF grant CCF-1910873}

\author{Riko Jacob}{IT University of Copenhagen, København S, Denmark}{rikj@itu.dk}{https://orcid.org/0000-0001-9470-1809}{Part of this work done during the second Hawaiian workshop on parallel algorithms and data structures, University of Hawaii at Manoa, Hawaii, USA, NSF Grant CCF-1930579}

\authorrunning{M. Goswami and R. Jacob} %

\Copyright{Mayank Goswami and Riko Jacob} %

\ccsdesc[500]{Theory of computation~Abstract machines}
\ccsdesc[500]{Theory of computation~Sorting and searching}
\keywords{Sorting, Priced Information, Instance Optimality, Nuts and Bolts} %

\category{} %

\relatedversion{An extended abstract of this work appeared in APPROX 2024} %

\nolinenumbers %

\EventEditors{John Q. Open and Joan R. Access}
\EventNoEds{2}
\EventLongTitle{42nd Conference on Very Important Topics (CVIT 2016)}
\EventShortTitle{CVIT 2016}
\EventAcronym{CVIT}
\EventYear{2016}
\EventDate{December 24--27, 2016}
\EventLocation{Little Whinging, United Kingdom}
\EventLogo{}
\SeriesVolume{42}
\ArticleNo{23}
\let\vv\vec
\def\A{\mathcal{A}}

\def\C{\mathcal{C}}
\def\D{\mathcal{D}}

\def\L{\mathcal{L}}
\def\T{\mathcal{T}}
\def\I{\mathcal{I}}
\def\N{\mathcal{N}}
\def\M{\mathcal{M}}
\def\D{\mathcal{D}}

\def\Qpos{{\mathbb{Q}^{\ge 0}}}

\DeclareMathOperator{\OPT}{OPT}

\date{}

\newcommand{\old}[1]{{{}}}

\begin{document}

\maketitle

\begin{abstract}
    We generalize the classical nuts and bolts problem to a setting where the input is a collection of $n$ nuts and $m$ bolts, and there is no promise of any matching pairs. It is not allowed to compare a nut directly with a nut or a bolt directly with a bolt, and the goal is to perform the fewest nut-bolt comparisons to discover the partial order between the  nuts and bolts. We term this problem \emph{bipartite sorting}.

    We show that instances of bipartite sorting of the same size exhibit a wide range of complexity, and propose to perform a fine-grained analysis for this problem. We rule out straightforward notions of instance-optimality as being too stringent, and adopt a \emph{neighborhood-based} definition. Our definition may be of independent interest as a unifying lens for instance-optimal algorithms for other static problems existing in literature. This includes problems like sorting (Estivill-Castro and Woods, ACM Comput. Surv. 1992), convex hull (Afshani, Barbay and Chan, JACM 2017), adaptive joins (Demaine, L\'{o}pez-Ortiz and Munro, SODA 2000), and the recent concept of universal optimality for graphs (Haeupler, Hladík, Rozhoň, Tarjan and Tětek, 2023).
    
    As our main result on bipartite sorting, we give a randomized algorithm that is within a factor of $O(\log ^3 (n+m))$ of being instance-optimal w.h.p., with respect to the neighborhood-based definition.
    
    As our second contribution, we generalize bipartite sorting to DAG sorting, when the underlying DAG is not necessarily bipartite. As an unexpected consequence of a simple algorithm for DAG sorting, we rule out a potential lower bound on the widely-studied problem of \emph{sorting with priced information}, posed by (Charikar, Fagin, Guruswami, Kleinberg, Raghavan and Sahai, STOC 2000). In this problem, comparing keys $i$ and $j$ has a known cost $c_{ij} \in \mathbb{R}^+ \cup \{\infty\}$, and the goal is to sort the keys in an instance-optimal way, by keeping the total cost of an algorithm as close as possible to $\sum_{i=1}^{n-1} c_{x(i)x(i+1)}$. Here $x(1) < \cdots < x(n)$ is the sorted order. While several special cases of cost functions have received a lot of attention in the community, no progress on the  general version with arbitrary costs has been reported so far.  One reason for this lack of progress seems to be a widely-cited $\Omega(n)$ lower bound on the competitive ratio for \emph{finding the maximum}. This $\Omega(n)$ lower bound by (Gupta and Kumar, FOCS 2000) uses costs in $\{0,1,n, \infty\}$, and although not extended to sorting, this barrier seems to have stalled any progress on the general cost case. 
    
    We rule out such a potential lower bound by showing the existence of an algorithm with a $\widetilde{O}(n^{3/4})$ competitive ratio for the $\{0,1,n,\infty\}$ cost version. This generalizes the setting of \emph{generalized sorting} proposed by (Huang, Kannan and Khanna, FOCS 2011), where the costs are either 1 or infinity, and the cost of the cheapest proof is always $n-1$.
\end{abstract}

\section{Introduction}

The classic nuts-and-bolts problem, originally mentioned as an exercise in \cite{rawlins1992compared}, asks: a disorganized carpenter has $n$ nuts and $n$ bolts, and there is (perfect) matching bolt for every nut. The only allowed comparison is between a nut and a bolt, and the result of such a comparison is either $<,=$ or $>$. The goal is to find the matching without comparing two nuts or two bolts to each other. A simple Quicksort type algorithm can be shown to solve this problem in optimal $O(n \log n)$ comparisons with high probability: Pick a random nut, compare to all bolts, find the matching bolt, and compare that bolt to all nuts. The problem is now partitioned into two subproblems with the match at the boundary; recurse. In a later work~\cite{alon1994matching}, Alon, Blum, Fiat, Kannan, Naor and Ostrovsky 
 developed a deterministic Quicksort-type algorithm that uses expanders and performs $O(n \text{ polylog } n)$ comparisons. This was later improved to an optimal $O(n \log n)$ comparisons by Koml\'{o}s, Ma, Sz\'{e}meredi \cite{komlos1998matching}, by performing substantial modifications on AKS sorting.

Our starting point is a remark by Koml\'{o}s, Ma and Sz\'{e}meredi: ``In particular, the fact that we can sort the nuts and bolts at all relies on the fact that there is a match between them.'' Indeed, if there is no matching (all comparisons come out $<$ or $>$), one realizes that the above randomized Quicksort based algorithm fails, as there is no partitioning into subproblems. The only case where sorting without a matching is possible is when nuts and bolts alternate in the final sorted order. Call this the perfectly interleaved case.
Koml\'{o}s, Ma and Sz\'{e}meredi observed (in a private communication with Aumann) that their AKS sorting-based algorithm sorts the nuts and bolts using $O(n \log n)$ comparisons in this setting\footnote{For a simple randomized algorithm that does the same, see Appendix~\ref{sec:backbonesort}.}. However, a general instance may not be perfectly interleaved, and this setting was left open.

\vspace{1mm}\noindent\textit{\underline{Generalized Nuts and Bolts}.} We focus on the problem alluded to by Koml\'{o}s, Ma, and Sz\'{e}meredi: what if the carpenter is completely disorganized, and has an unequal collection of nuts and bolts, without any matching pairs\footnote{It can be observed that having some matchings in the input only makes it easier to solve.}? 
That is, assume the carpenter has a set $R$ of $n$ nuts and a set $B$ of $m$ bolts, and is only allowed to compare a nut to a bolt, and the result is either $<,or >$. 
Unless $m=n$ and we are in the perfectly interleaved case, sorting $R \cup B$ is not possible: 
there could be two (or more) nuts (resp. bolts) that compare the same way to all the bolts (resp. nuts). 
A natural goal for the carpenter now is to ``sort as much as you can'', i.e., partition the set of nuts $R$ into subsets $R_1,R_2,\ldots$ such that for any $r,r' \in R_{i}$ and any $b \in B$, $r$ and $r'$ have the same order with $b$ (either both are smaller, or both are larger), and vice-versa (see figure~\ref{fig:stripes}). We term this generalization of nuts and bolts as \emph{bipartite sorting}: given the complete bipartite graph $G=(R \cup B,E)$, the goal is to discover the orientation\footnote{An edge $e=(u,v)$ in $G$ is oriented as $\vec{uv}$ if $u <v$ in the underlying DAG.} on all the edges in $E$ by querying as few of them as possible. 

\begin{figure}[htb]
    \centering
    \includegraphics[width=14cm]{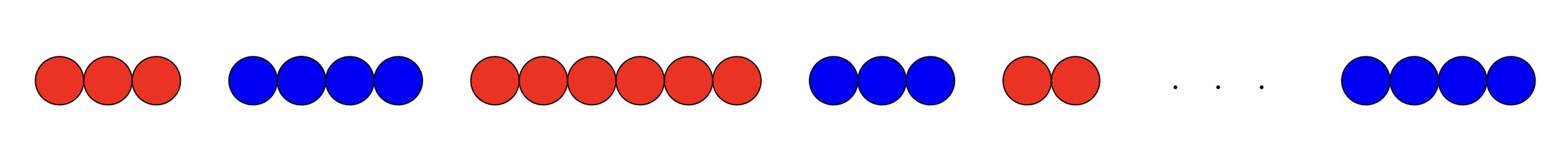}.
    \caption{An example output to an instance of the bipartite sorting problem. Continuous runs of incomparable nuts and bolts are called ``stripes''.}
    \label{fig:stripes}
\end{figure}

\noindent\textit{The need for fine-grained analysis} If all nuts are smaller than all bolts, it is clear that $nm$ comparisons are needed by any algorithm, and obviously this number of comparisons suffices for any instance. 
Recall that the perfectly interleaved case can be solved (in fact completely sorted) much faster with only $O(n \log n)$ comparisons. Assuming for simplicity that $m=\Theta(n)$, the inherent complexities of instances range in $[\tilde{O}(n),\Omega(n^2)]$. Is there a way to define instance-optimality for bipartite sorting that captures its variety of underlying instances, and is there a good instance-optimal algorithm?

It is worthwhile to imagine the DAGs for the above two instances: if all nuts are smaller than all bolts, the DAG $\vec{G}=(V=R \cup B,\vec{E})$ will have all edges oriented from $R$ to $B$, whereas the DAG in the perfectly interleaved case will be the transitive closure of the oriented edges in the \emph{directed Hamiltonian path corresponding to the sorted order of the nuts and the bolts}. A \emph{transitive reduction} of a DAG is the fewest number of oriented edges that by transitivity imply all other orientations. For the first instance, the transitive reduction has size $nm$, whereas for the second case, it has size $n-1$, and therefore if we ignore $\log ^{O(1)} n$ factors, for both instances, the sizes of their transitive reductions matches their complexity.
This immediately suggests a parametrization similar to an output-sensitive setting:
the number of comparisons of a good algorithm should be close to (say, within log factors of) the size of the transitive reduction of the underlying DAG.

However, the following instance dashes all hopes of an algorithm that performs only roughly as many queries as there are edges in the transitive reduction: 
$n-1$ nuts are all smaller than a special bolt $b$, which is smaller than a special nut $r$, which is smaller than all other $n-1$ bolts. We call this the \emph{``one-inversion''} instance. Even though the transitive reduction here has size $2n-1$, any algorithm must perform $n^2$ comparisons to find the hidden special pair $(r,b)$. Thus the gap between the transitive reduction and the inherent complexity of this instance is $\Omega(n)$, which is as large as the gap between the complexity of any two instances. Observe that the DAG for this instance is in some sense just ``one-flip-away'' from the all-nuts-smaller-than-all-bolts DAG, a phenomenon that will be important later. Given that the transitive reduction fails to capture the instances, we ask:

\vspace{-1mm}\begin{center}
    \textit{Is there another meaningful way to define instance-optimality for bipartite sorting that captures its variety of underlying instances, and is there a good instance-optimal algorithm?}
\end{center}

\begin{figure}
\begin{adjustbox}{center}
\begin{tikzpicture}[every text node part/.style={align=center}]

 \node(1) at (3,1) [rounded rectangle, draw, minimum height=1cm]{Arbitrary $f(x_1,\cdots,x_n)$ with Priced Information \cite{charikar2002query}};
  \node(19) at (-2,-2) [rounded rectangle, draw, minimum height=0.75cm]{Searching, Other Problems\\ \cite{kaplan2005learning,deshpande2014approximation,onak2006generalization}\\ \cite{mozes2008finding,bender2021batched}};
 \node(2) at (3,-2) [rounded rectangle, fill=gray!20, draw, minimum height=0.75cm] {DAG Sorting,\\ Arbitrary Costs, HA$\times$};
 \node(3) at (3,-6) [rounded rectangle, draw, minimum height=0.75cm] {Sorting with Priced\\ Information\cite{charikar2002query}, HA\Checkmark};
  \node(4) at (3,-11) [rounded rectangle, fill=gray!20, draw, minimum height=0.75cm] {Sorting with costs\\$\{0,1,n,\infty\}$ HA\Checkmark};
  \node(5) at (3,-15) [rounded rectangle, draw, minimum height=0.75cm] {Generalized Sorting\cite{6108244}, \\ Costs $\in \{1,\infty\}$, HA\Checkmark};

 \node(6) at (7,-6) [rounded rectangle, fill=gray!20, draw, minimum height=0.75cm] {Bipartite\\ Sorting, \\ HA$\times$};
 \node(7) at (7,-13) [rounded rectangle, draw, minimum height=0.75cm] {Bichromatic \cite{charikar2002query}\\ Sorting~\cite{bichromaticITCSarxiv}, HA\Checkmark};
 \node(8) at (3,-17.5) [rounded rectangle, draw, minimum height=0.75cm] {Unit-Cost Sorting};
 \node(9) at (-.8,-11) [rounded rectangle, draw, minimum height=0.75cm] {Structured Costs\\ Sorting\cite{gupta2001sorting}, HA\Checkmark};
 \node(10) at (11,-14) [rounded rectangle, fill=gray!20, draw, minimum height=0.75cm] {Perfectly \\ Interleaved, HA\Checkmark};
 \node(11) at (11,-10) [rounded rectangle, fill=gray!20, draw, minimum height=0.75cm] {Not Perfectly\\Interleaved, HA$\times$};
 \node(12) at (8.5,-16) [rounded rectangle, draw, minimum height=0.75cm] {Nuts and Bolts\cite{komlos1998matching, alon1994matching}};
 \node(13) at (-2.75,-6) [rounded rectangle, fill=gray!20, draw, minimum height=0.75cm] {DAG Sorting, \\Costs $\{0,1,\infty\}$ \\ HA$\times$};
 \node(14) at (11,-1) [rounded rectangle, draw, minimum height=0.75cm]{Finding Maximum \\Arbitrary Costs \cite{charikar2002query,gupta2001sorting}};
 \node(15) at (13,-4) [rounded rectangle, draw, minimum height=0.75cm]{Finding Maximum \\Costs $\{0,1,n,\infty\}$\cite{charikar2002query,gupta2001sorting}};

\draw [thick, ->] (1)--(2);
\draw [thick, ->] (2)--(3);
\draw [thick, ->] (3)--(4);
\draw [thick, ->] (4)--(5);
 \draw [thick, ->] (1)--(14);
 \draw [thick, ->] (14)--(15);
 \draw [thick, ->] (2)--(6);
 \draw [thick, ->] (2)--(13);
 \draw [thick, ->] (6)--(10);
 \draw [thick, ->] (6)--(11);
 \draw [thick, ->] (10)--(12);
\draw [thick, ->] (12)--(8);
\draw [thick, ->] (3)--(9);
\draw [thick, ->] (3)--(7);
\draw [thick, ->] (7)--(10);
\draw [thick, ->] (5)--(12);
\draw [thick, ->] (1)--(19);
\draw [thick, dotted, ->] (15) --node[near start, above right, sloped]{\cancel{\LARGE{$\Omega$}}} (4);
\draw [thick, dotted, ->] (13) --node[near start, above right, sloped]{\LARGE{$O$}} (4);
\draw [thick,->] (13)  --(-2.9,-15) |-(5);
\draw [thick, dotted, <-](6) --node[midway,sloped]{\Large{InversionSort}} (7);

\end{tikzpicture} 
\end{adjustbox}
\vspace*{\fill}

\caption{\small{The landscape of sorting with priced information. Solid arrows go from a problem to its special case. HA\Checkmark indicates that the Hamiltonian path is assumed to exist and HA$\times$ indicates that a Hamiltonian path may not exist. Problems shaded in gray are introduced or studied in this paper for the first time. Dotted arrows highlight our results, arrows with $O$ show algorithms carrying over from one problem to another, and \cancel{$\Omega$} show lower bounds not carrying over.}}
\label{fig:landscape}

\end{figure}
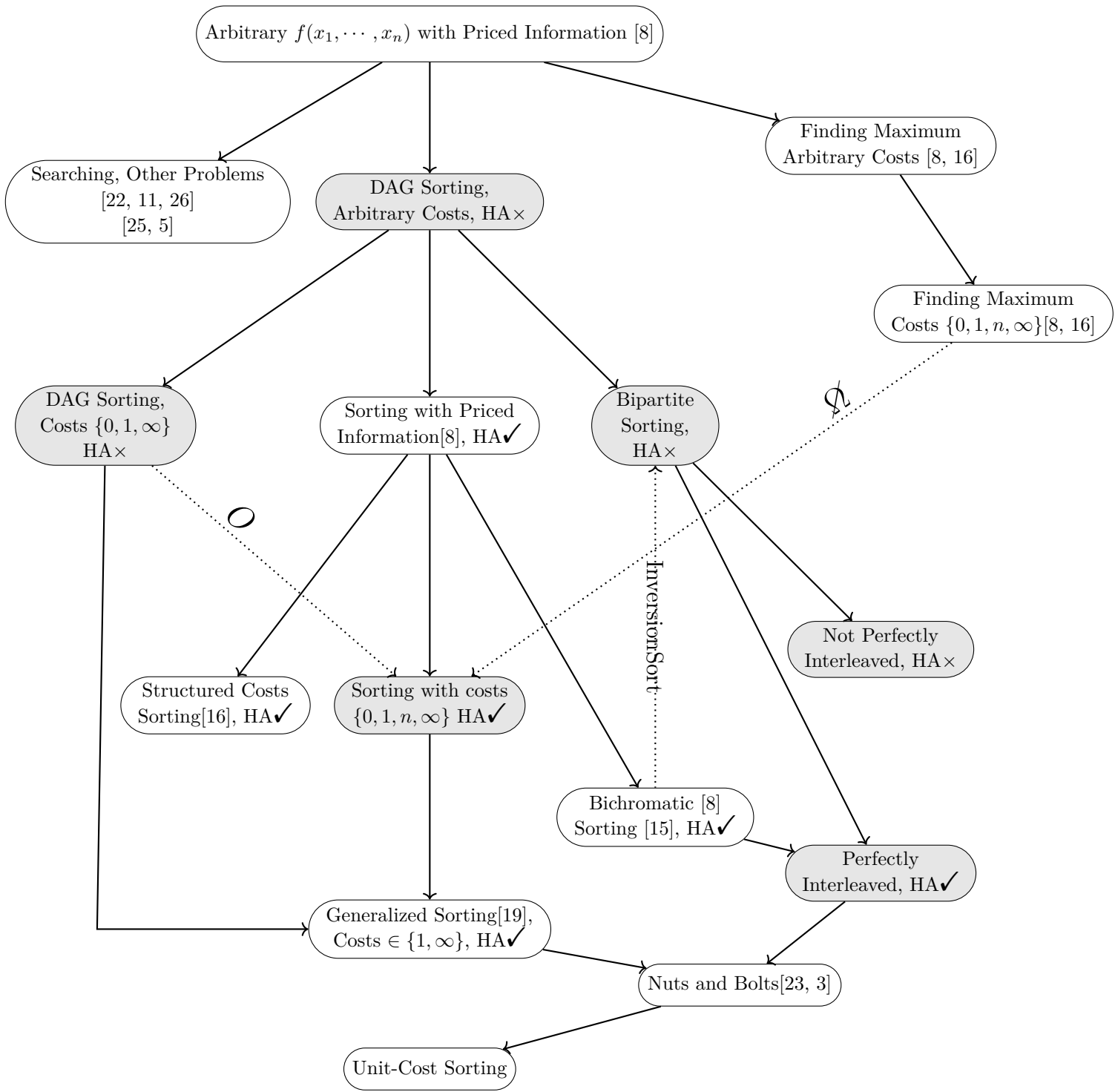
 
\section{Our Results and Technical Overview}

In this paper, we answer the above question in the affirmative, and in the process of doing so, make unexpected progress on a widely-studied problem called generalized sorting (Huang, Kannan and Khanna \cite{6108244}), which in turn is a special case of the sorting with priced information problem introduced by Charikar et al. \cite{charikar2002query}. We explain this connection first before stating our results.

\vspace{2mm}\noindent\textbf{DAG Sorting.} One can generalize bipartite sorting to DAG sorting where the set of allowed comparisons is represented by an arbitrary (not necessarily bipartite) graph $G=(V,E)$. The goal is still to discover all orientations in the underlying DAG $\vec{G}=(V,\vec{E})$ or equivalently, its transitive reduction, by querying as few edges of $G$ as possible, where a query of an edge $(u,v) \in E$ returns $<$ or $>$. Edges not present in $E$ cannot be queried.

It turns out that under the promise that $\vec{G}=(V,\vec{E})$ has a directed Hamiltonian path, DAG sorting is exactly equivalent to generalized sorting. 
Thus both bipartite sorting\footnote{And by transitivity, sorting nuts and bolts when perfectly interleaved, matching nuts and bolts, and classical sorting.} and generalized sorting can be viewed as special cases of DAG sorting. In fact, one can go a step further and assign non-negative costs to the edges in $E$ in DAG sorting, and now ask for algorithms that find the transitive reduction with the cheapest cost. If the underlying DAG has a Hamiltonian path, this is exactly the problem of sorting with priced information. 

\subsection{Bipartite Sorting}

As our main result, we define a meaningful notion of \emph{instance optimality} for bipartite sorting, and give an algorithm that is instance optimal up to a factor of $O(\log^3(n))$. We found this (definition, algorithm) result surprising because it shows that in some sense, the ``one-flip'' phenomenon (recall the ``one-inversion'' instance) mentioned in the introduction is the \emph{only} obstruction to achieving (almost) instance optimality. We believe our definition may be of independent interest as it unifies several previously-studied notions of instance optimality. 

\subsubsection{Defining Instance-Optimality}  

The transitive reduction of $\vec{G}$, denoted as $\vec{T}$, can be thought of as the ``cheapest proof'' of the underlying DAG being $\vec{G}$. 
Such a comparison to the cheapest proof has historically been very useful in defining instance optimality. Indeed, for the problems of generalized sorting and sorting with priced information, when the transitive reduction equals the directed Hamiltonian path, the cost of this directed Hamiltonian path is the measure of instance-optimality (and the factor by which an algorithm exceeds it is called \emph{competitive ratio} in the original work by Charikar et al.\cite{charikar2002query}). The cheapest proof appears again in the work by Demaine, L\'{o}pez-Ortiz and Munro \cite{demaine2000adaptive}], who study the problems of comparison-based set unions, intersections and differences. 
However in our setting, as the one-inversion instance shows, the size of the transitive reduction $|\vec{T}|$ is \emph{too strong a benchmark} to compare the performance of an algorithm to.  

An important work on instance optimality deviating from a comparison to cheapest proof is by [Afshani, Barbay and Chan,\cite{afshani2017instance}] which studies instance optimal algorithms for the convex hull problem. This was also adopted by [Cardinal, Dallant and Iacono \cite{cardinal2021instance}], who studied bichromatic rectangular visibility. When $x$ is an input sequence of points, \cite{afshani2017instance} and \cite{cardinal2021instance} define $OPT(x)$ to be the runtime of an algorithm that is \emph{order-oblivious}, i.e., OPT's code can depend on the set of elements in $x$, but not on their order. The fact that \emph{some} restriction on OPT is needed (at least for static problems) follows because comparing the runtime of an algorithm $A$ on an input $x$ to an algorithm (OPT) that is ``tailor-made'' for an input $x$ is too strong: if OPT's code can depend on $x$, $OPT(x)$ can be very small compared to $A(x)$. In the convex hull problem $OPT(x)$ would just be $O(n)$ as OPT only needs to read the input\footnote{[Afshani, Barbay and Chan \cite{afshani2017instance}] state: ``For example, consider the 2D convex hull problem, which has $\Theta(n \log n)$ worst-case complexity in the algebraic computation tree model: for every input sequence of $n$ points, one can easily design an algorithm $A$ (with its code depending on the input sequence) that runs in $O(n)$ time on that particular
sequence, thus ruling out the existence of an instance-optimal algorithm.'' 
}, and for DAG/bipartite sorting if OPT knows the underlying DAG $\vec{G}$ then (since we only count query complexity) OPT equals zero!

We posit a \emph{``neighborhood-based''} approach to defining notions of instance optimality. In a nutshell, we observe that most definitions for instance optimality boil down to choosing the ``right'' neighborhood. The smaller the neighborhood is, the stronger the notion of instance-optimality, and the harder it is to attain. If the neighborhood is the set of all instances, we are back to worst-case analysis. Thus the art is in choosing the smallest meaningful neighborhood that still allows for instance optimal algorithms. This general step-by-step process of increasing the neighborhood until one sees hope for instance-optimality reveals quite a bit about the fine-grained nature of a problem. Perhaps this way of looking at instance-optimality was already known, but since we have not seen it in print, we show (in Appendix~\ref{app:framework}) how most of the existing works on instance-optimality, ranging from the classical works on adaptive sorting and presortedness \cite{AdaptiveCastroWood92}, to later works by Demaine et al. \cite{demaine2000adaptive} and Afshani et al.\cite{afshani2017instance}, and including the very recent work on universal optimality \cite{haeupler2023universal}, all agree with this paradigm . 

We define a neighborhood of the underlying DAG $\vec{G}$, denoted as $\N(\vv{G})$. See Definition~\ref{introdef} for a precise definition but intuitively, this is the set of all DAGs that are either isomorphic, or ``one-flip'' away from $\vec{G}$. Define the runtime $T_{A}(\vv{G})$ of a randomized algorithm $A$ on an instance $\vv{G}$ as its expected comparison-cost on instance~$\vv{G}$. %
Let $\C(\vv{G})$ be the set of randomized (Las Vegas) algorithms that are correct for all instances in~$\N(\vv{G})$.
Now define \[\OPT(\vv{G}) = \inf_{A\in \C(\vv{G})} \max_{\vv{G}'\in \N(\vv{G}) } T_{A}(\vv{G}').\]

\noindent For $\alpha \geq 1$, we say an algorithm $A$ is $\alpha$-instance optimal if for every instance $\vv{G}$, $T_{A}(\vv{G}) \leq \alpha \OPT(\vv{G})$. 
Does there exist an algorithm achieving $\alpha=O(1)$, or $\alpha= \log^{O(1)} n$?

\subsubsection{InversionSort: An Almost Instance Optimal Algorithm}

The algorithm we investigate for bipartite sorting is a variant of an algorithm InversionSort that was recently presented by us (Goswami and Jacob~\cite{bichromaticITCSarxiv}) for the version when there is a sorted order on the nuts and bolts, and nuts can be compared to nuts (at a cost $\alpha>1$) and bolts can be compared to bolts (at a cost $\beta>1$). In bipartite sorting, two nuts (or two bolts) cannot be compared to each other, i.e., $\alpha=\beta=\infty$, but due to the similarity to the previous algorithm we call the algorithm for bipartite sorting InversionSort too. Note that in \cite{bichromaticITCSarxiv} the previous algorithm was compared to the cost of the Hamiltonian, whereas now not only may the Hamiltonian cease to exist, its natural counterpart (the transitive reduction) is a weak lower bound.

As a first try, let us see what goes wrong when we apply the simple randomized QuickSort algorithm for nuts and bolts to bipartite sorting. We pick a random nut $r$ and use it to pivot the bolts obtaining two sets $B_{<r}$ and $B_{>r}$. Since no match was obtained, in the alternating step we can select a random bolt $b$, say from $B_{>r}$, and pivot the nuts, obtaining $R_{<b}$ (containing $r$) and $R_{>b}$. Instead of the two perfectly-partitioned subproblems obtained in nuts and bolts, we unfortunately now have three subproblems: $(R_{<b},B_{<r})$, $(R_{<b},B_{>r})$, and $(R_{>b},B_{>r})$. 

\emph{If} the nuts and bolts were perfectly interleaved, we show (Appendix~\ref{sec:backbonesort}) that a BFS-style Quicksort algorithm that we call BackboneSort sorts using $O(n \log n)$ comparisons. BackboneSort works in a sequence of alternating phases, a nut-phase and a bolt-phase. In a nut phase it tries to refine the nut subproblems and makes progress over the three subproblems above in a round-robin fashion: for example, to make progress in $(R_{<b},B_{<r})$, it would select a random bolt in $B_{<r}$ and pivot $R_{<b}$, and vice-versa in a bolt-phase. 

Unfortunately, in a general instance of bipartite sorting the nuts and bolts are not perfectly interleaved everywhere, and it turns out that BackboneSort can perform badly (see Theorem~\ref{thm:balloon}). The reason for this behavior arises from selecting a ``random'' pivot in the subproblem: for the perfectly interleaved instance, its sub-instances are also perfectly interleaved (in particular all three subproblems $(R_{<b},B_{<r})$, $(R_{<b},B_{>r})$, and $(R_{>b},B_{>r})$ can be shown to be roughly a constant fraction of the original problem with good probability) , and then such a random pivot is guaranteed to be good, just like in randomized Quicksort. However, if the underlying instance is lop-sided (think of many nuts in $R_{<b}$ sandwiched between two bolts in $B_{<r}$), a random bolt is not a good pivot. To add to this complication, an instance may be lop-sided in one region and perfectly interleaved in another, and an instance-optimal algorithm should ideally be able to detect such a situation.

InversionSort proceeds similarly to BackboneSort, but instead of selecting random pivots, it performs a random comparison: when trying to make progress in $(R_{<b},B_{<r})$, it will look for ``inversions'', defined as a pair $(r,b)\in R_{<b} \times B_{<r}$ such that $b<r$. When it finds such a pair, it uses them as pivots. The intuition is that regions having ``easy'' subinstances (like perfectly interleaved ones) will resolve fast, whereas those that are lop-sided, like all nuts less than all bolts, will take longer, but this is a necessary task in this sub-region.

The following theorem quantifies the performance of InversionSort. Here, $N=n+m$ and Definition~\ref{introdef} is our precise definition of instance optimality, sketched in the previous section.

\begin{restatable}{theorem}{bipartitetheorem}[Instance Optimality of InversionSort]\label{thm:bipartiteinversionsort}
    There exists a constant $c>0$, such that for every instance~$\I$, the cost of InversionSort on~$\I$ is, with probability at least $1-1/N$, at most $c(\log N)^3 \OPT(\I)$, where $\OPT(\I)$ is as in Definition~\ref{introdef}.
\end{restatable}

Thus InversionSort is $O(\log^3 N)$-instance optimal, with respect to a natural notion of instance-optimality that accounts for the one-flip neighborhood of the underlying DAG. We conjecture that InversionSort is actually $O(1)$-instance optimal with this notion\footnote{As evidence, we prove in Theorem~\ref{thm:InvsortOnInterleaved} in Appendix~\ref{sec:InvsortOnInterleaved} that InversionSort solves the perfectly interleaved instance in an optimal $O(n \log n)$ comparisons.}.

\subsection{Unexpected Result: Generalized Sorting and Sorting with Priced Information}

While Theorem~\ref{thm:bipartiteinversionsort} gives an algorithm for bipartite sorting that is instance-optimal, what can we say for DAG sorting? Unfortunately proving a polylog-instance optimality guarantee seems out of reach for now. The reason is that even the ``sortable'' case of DAG sorting, where the DAG has a directed Hamiltonian path, has a current best bound of $\tilde{O}(n^{1.5})$ comparisons, or in other words, is a factor $\tilde{O}(\sqrt{n})$ away from the Hamiltonian cost (which is $n-1$). This results from an interesting randomized algorithm by Huang, Kannan and Khanna~\cite{6108244} that uses the work by Kahn and Linial~\cite{kahn1991balancing} on balancing extensions via the Brunn-Minkowski theorem.

For a DAG that is not sortable, we extend the algorithm of Huang, Kannan and Khanna to give an algorithm that performs $\tilde{O}(\min(wn^{1.5},n^2))$ comparisons and outputs the transitive reduction of $\vec{G}$ (see Theorem~\ref{get01dag}). Here $w$ denotes the width of the DAG, which is the size of its largest antichain (a set of incomparable elements). We now point out an unexpected consequence of this result.

The problem of \emph{sorting with priced information} introduced by Charikar, Fagin, Guruswami, Kleinberg, Raghavan and Sahai, \cite{charikar2002query}, is a  generalization of the classical, unit-cost, comparison-based sorting, defined as follows. The input is a weighted undirected graph $G$ on $n$ vertices, with the cost $c_{ij} \in \mathbb{R}_{\geq 0}$ on the edge $e_{ij}$ indicating the cost to compare keys (represented by vertices) $v_i$ and $v_j$. As before, edges not in $G$ have cost $\infty$ and cannot be queried, and a  query on an edge $e_{ij}$ reveals if $v_i <v_j$ (indicated as $\vec{e}_{ij}$) or $v_{i} >v_{j}$ (indicated as $\vec{e}_{ji}$).

Since the hidden Hamiltonian path $\mathcal{H}$ is the cheapest proof, its cost which equals $\sum_{i=1}^{n-1} c_{x(i)x(i+1)}$ is a lower bound. Here $x(1) < \cdots < x(n)$ is the sorted order. \cite{charikar2002query} propose finding algorithms that come as close to the cost of $\mathcal{H}$ as possible. The \emph{competitive ratio} is defined as the ratio of the cost of the algorithm to the cost of $\mathcal{H}$, and the goal is to find an algorithm with small competitive ratio.

 Several special cases of cost functions have been studied, but for arbitrary costs, almost nothing is known about the above problem. Note that the $\tilde{O}(n^{1.5})$ result by Huang, Kannan and Khanna~\cite{6108244} works when all costs are either $1$ or $\infty$. What about the version with arbitrary costs? Many works state that the general-cost version is ``arbitrarily bad'' (Huang, Kannan and Khanna \cite{6108244}), ``bleak'' or ``hopeless'' (Gupta and Kumar \cite{gupta2001sorting}). The only evidence for this is an $\Omega(n)$ lower bound on the competitive ratio of any algorithm that \emph{finds the maximum}. There is an example where the costs are either $0$, $1$, $n$ or $\infty$, and one can show that any algorithm that finds the maximum element $m$ must have cost $\Omega(n)$ times that of the cheapest proof (of $m$ being the maximum) (Charikar et al. \cite{charikar2002query}, Hartline et al. \cite{hartlinepersonal}, Gupta and Kumar \cite{gupta2001sorting}). While it certainly provided intuition, the $\Omega(n)$ lower bound for maximum (with costs in $\{0,1,n,\infty\}$) was never extended to sorting. 
 
 This makes the $\{0,1,n,\infty\}$ version interesting due to three reasons:
 \begin{compactitem}
     \item do instances in this cost regime contain an $\Omega(n)$ lower bound for sorting too?
     \item it is the natural step-up from generalized sorting with costs in $\{1,\infty\}$, and 
     \item this is the first instance with forbidden comparisons that requires an \emph{instance-specific analysis} for the competitive ratio. For generalized sorting and for stochastic sorting\footnote{Also initiated by Huang, Kannan and Khanna~\cite{6108244}, this is the version where the input graph $G$ is random}, the cost of the Hamiltonian is always $n-1$, so one only has to bound the cost of an algorithm, without worrying about the underlying instance. However, when costs are in $\{0,1,n,\infty\}$ the cost of the Hamiltonian can range from $0$ to $n(n-1)$, and so an algorithm must adapt to the underlying instance.
 \end{compactitem} 

Our second theorem addresses this cost version of sorting with priced information, showing that it cannot be the $\{0,1,n,\infty\}$ version that makes sorting hopeless!

\begin{restatable}{theorem}{resultongeneralsorting}\label{mainresult1}
   Consider the problem of sorting when every comparison has a cost in $\{0,1,F, \infty\}$, for any $F \geq n^{3/4}$. There exists a polynomial time randomized algorithm whose competitive ratio is $\widetilde O(n^{3/4})$, with high probability. 
\end{restatable}

The main ingredient in the proof of this theorem is the aforementioned $\tilde{O}(\min(wn^{1.5},n^2))$ comparisons algorithm for DAG sorting. Even though DAG sorting does not promise a Hamiltonian, it turns out to be useful because it can be used as a subroutine in a ``greedy'' algorithm for the $\{0,1,n,\infty\}$ cost version: obtain with cheapest cost the partial DAG formed by all costs $0$ and $1$ comparisons.

\vspace{2mm}\noindent\textbf{Organization:} We state our problems precisely in Section~\ref{sec:problemDeg}. This is followed by our result on bipartite sorting (Theorem~\ref{thm:bipartiteinversionsort}) in Section~\ref{sec:bipartitemain}, and our result on sorting with priced information (Theorem~\ref{mainresult1}) in Section~\ref{sec:generalsorting}.

\section{Problem Definitions}\label{sec:problemDeg}

We formally define the problem of bipartite sorting first, and then the problems of DAG sorting and sorting with priced information. We invite the reader to see Figure~\ref{fig:landscape} in the appendix for the landscape of these problems, their relations to each other, and how our results fit in this landscape.

\begin{definition}[Bipartite Sorting]\label{def:bipartitesorting}

Input: A complete bipartite undirected graph $G$ of unit costs on $V= R \cup B$. Only edges in $G$ can be queried (at unit cost), and querying an undirected edge $(u,v)$ has one of two outcomes, $u<v$ (implying $\vec{uv} \in \vec{G}$) or $u>v$ (implying $\vec{vu} \in \vec{G}$). $N:=|V|$.

The \textit{instance} of bipartite sorting is defined by a partition of $R$ (reds) and $B$ (blues) into stripes  $(S_1,\cdots,S_k)$ (Figure~\ref{fig:stripes}), i.e., by the relative order\footnote{$S_1$ is the set of all sources in the DAG $\vv G$; since red and blue elements can be compared with unit-cost, $S_1$ must necessarily be of one color. $S_i$ can be iteratively defined as the set of sources in $\vv G$ after all stripes $1$ to $i-1$ have been deleted.} between the reds and the blues. Note that $S_{i}$s are unordered sets. The DAG $\vec G$ has, for all $1 \leq i \leq k-1$, an edge from every element in $S_i$ to every element of the other color greater than it, i.e., to every element in $\cup_{\ell \geq 0} { S_{i+(2\ell+1)}}$.

\noindent Output: The sequence of stripes $(S_1,\cdots,S_{k})$ (see figure~\ref{fig:stripes}). Equivalently, the transitive reduction of $\vec G$. 
\end{definition}

\noindent\textbf{DAG Sorting.} Let $\mathcal{P}$ denote a partial order on a set of $n$ elements, and let $\vv{P}$ denote the transitively closed DAG on $V=\{v_1,\cdots,v_{n}\}$ indicating all order relations in $\mathcal{P}$. That is, the vertex $v_{i}$ is identified with element $i$, the edge $\overrightarrow{e_{ij}}$ between vertices $i$ and $j$ exists if $v_i < v_j$, the edge $\overleftarrow{e_{ij}}$ exists if $v_i > v_j$, and no edge between $v_i$ and $v_j$ exists if elements $i$ and $j$ are incomparable. For convenience, let $P$ denote the set of (undirected) edges in $\vv{P}$ without their directions. Let $w$ denote the width of $\vv{P}$.

\begin{definition}[Implied and essential edge, transitive reduction \cite{aho1972transitive}]\label{def:transitiveReduction}
    Given a DAG~$\vv G=(V,\vv E)$, an edge $(u,v)\in \vv E$ is \emph{implied}, if there is a directed path in $\vv G$ from $u$ to~$v$.
    Otherwise, $(u,v)$ is \emph{essential}. 
    The set of essential edges is called the \emph{transitive reduction} of $\vv G$.
\end{definition}

\noindent Note that for every implied edge $(u,v)\in \vv E$, there is a directed $u$ to $v$ path of essential edges in~$\vv G$. It turns out that the transitive reduction of a DAG is unique \cite{aho1972transitive}. Let $\vv{\mathcal{T}}$ denote the transitive reduction of $\vv{P}$.  Let $\mathcal{T}$ denote the undirected version of $\vv{\mathcal{T}}$.

\begin{definition}[DAG Sorting, arbitrary costs]\label{def:dagsorting}

Input: An undirected graph $G=(V,E \subset P)$ with costs $c_{ij} \in \mathbb{R}_{\geq 0}$ on edges. An oracle that answers, given an undirected edge $e_{ij}$ of $G$, its orientation in $\vv{P}$. 

Promise: $\mathcal{T} \subset E$. That is, the queryable edges contain the edges of the transitive reduction.

Output: $\vv{\mathcal{T}}$.

Cost of an algorithm $\mathcal{A}$: The total cost of the edges queried by $\mathcal{A}$ on the instance $\vv{P}$. This will be denoted by $\text{cost}(\mathcal{A},\vv{G})$, where $\vv{G}$ is the directed version of $G$ (also referred to as the instance from now on), and contains all the information about $\vv{P}$.

When $c_{ij}=1$ for all $i,j$, we call the problem simply DAG Sorting.

\end{definition}

\noindent We will only care about the query cost of the algorithm, and while there may be a compact representation of $\vv{\mathcal{T}}$, we ask the algorithm to output $\vv{\mathcal{T}}$ in its entirety for simplicity.

DAG sorting generalizes sorting with priced information, which we describe next. This problem was introduced by \cite{charikar2002query} in the broader context of querying with priced information, where one wants to compute a function $f$ of $n$ inputs, and querying an input has a certain cost associated to it. The competitive ratio is defined as the (worst case) ratio of the cost of the query strategy to the \textit{cost of the cheapest proof of $f$}. This work initiated a multitude of papers on priced information, studying problems like learning with attribute costs \cite{kaplan2005learning}, stochastic boolean function evaluation \cite{deshpande2014approximation}, searching on trees \cite{onak2006generalization,mozes2008finding}, and priced information in external memory \cite{bender2021batched}.

\begin{definition}[Sorting with Priced Information \cite{charikar2002query}]
Sorting with priced information is a special case of DAG sorting, when the partial order $\mathcal{P}$ is a total order. In this case, $\vv{P}$ is a tournament and therefore $P$ is a complete graph. $\vv{\mathcal{T}}$ is simply a directed Hamiltonian path $\mathcal{H}$ in $\vv{P}$. The input $G$ is any graph on $n$ vertices containing the edges of $\mathcal{H}$ without their directions, and the output is $\mathcal{H}$. The total cost of the edges queried by an algorithm $\mathcal{A}$ on the instance $\vv{G}$ will be denoted by $\text{cost}(\mathcal{A},\vv{G})$.
\end{definition}

\noindent\textbf{Competitive ratio for sorting with priced information} The \textbf{competitive ratio} of $A$ (as defined in \cite{charikar2002query}) is $\rho(\mathcal{A},n)= \max_{\vv{G}} \text{cost}(A,\vv{G})/\text{cost}(\mathcal{H})$, where the maximum is taken over all instances $\vv{G}$ of $n$ vertices with a total order. The goal is to find sorting algorithms with small competitive ratio. Here $\text{cost}(\mathcal{H})$ is considered as a proxy for the complexity of the instance $\vv{G}$, as it is the cheapest proof. It certainly is a valid lower bound, for the edges on the Hamiltonian \textit{must} be queried by any algorithm.

\noindent For example, when $G$ is the complete unit-cost graph, MergeSort or QuickSort achieve a competitive ratio of $\Theta(\log n)$ (the latter w.h.p. if randomized). Similarly, when $G$ is unit-cost but not complete, the $\tilde{O}(n^{1.5})$ cost algorithm algorithm by [Huang, Kannan and Khanna \cite{6108244}] has a competitive ratio of $\tilde{O}(\sqrt{n})$. As mentioned, very little is known about the case when $G$ has arbitrary costs. 

We end with the remark that DAG sorting is closely related to a line of work initiated by [Faigle and Tur\'{a}n,\cite{faigle1988sorting}] called sorting a partial ordered set, or identifying a poset. This was followed up by several works such as \cite{daskalakis2011sorting} and \cite{dubhashi1993searching}. For a survey on this line of work that also includes generalized sorting, we refer the reader to \cite{cardinal2013generalized}.

\section{Results on Bipartite Sorting}\label{sec:bipartitemain}

This section is divided into three subsections. In Section~\ref{sec:defineio} we formally state our definition of instance-optimality. In Section~\ref{subsec:lower} we derive some lower bounds on OPT stemming from the definition of instance-optimality. Finally, in Section~\ref{subsec:inversionsortanalyze} we define InversionSort, and prove that it comes close to achieving instance-optimality (Theorem~\ref{thm:bipartiteinversionsort}) by charging the comparisons performed by InversionSort to the derived lower bounds.

\subsection{Defining Instance-Optimality}\label{sec:defineio}

As mentioned in the introduction, the following instance of bipartite sorting shows that comparing the cost of an algorithm to the transitive reduction is hopeless. 
\begin{definition}[One-inversion Instance]\label{def:one-flip} Let $G=(R,B,E)$ be the undirected complete bipartite graph on $|R|=|B|=n/2$. Pick an arbitrary $r \in R$, $b \in B$, let $R_{-r}=R \setminus \{r\}$ and $B_{-b}=B \setminus \{b\}$. Define a DAG $\vv{G}$ via its transitive reduction as follows.
\[\text{TR}(\vv{G})=    \{\vv{xb}: x \in R_{-r} \} \cup \{ \vv{br}\} \cup \{ \vv{ry}: y \in B_{-b} \} .\]
\end{definition}

The transitive reduction has size $O(n)$, but any algorithm must spend $\Omega(n^2)$ comparisons to identify $r$ and $b$. Thus the ``cheapest proof'' is too strong a benchmark. We now present our neighborhood-based approach for bipartite sorting. This approach is general, and in the Appendix~\ref{app:framework} we show how this neighborhood-based approach fits several works on instance-optimal algorithms for static problems, namely the works on adaptive sorting \cite{AdaptiveCastroWood92}, on set intersection, union and difference \cite{demaine2000adaptive}, and the recent work on universal optimality \cite{haeupler2023universal}. 

We start with small neighborhoods and gradually increase them until there is no immediate obstruction to instance-optimality.
Define $\N_{\A}(\vv{G})$ as the set of DAGs (Automorphic) isomorphic to $\vv{G}$ if all edges in $\vv{G}$ are unit-cost, and cost-isomorphic\footnote{Meaning there exists an isomorphism btween the DAGs that preserves the costs.} otherwise. Next, define the runtime of an algorithm on an instance $\vv{G}$ as its maximum comparison-cost on any instance in $\N_{\A}(\vv{G})$. Define $OPT_{\A}(\vv{G})$ as the smallest comparison-cost of any algorithm. 

For unit-cost sorting, $\vv{G}$ is a complete DAG, and $|\N_{\A}(\vv{G})|=n!$. 
It is easily seen that now $OPT(\vv{G}) =\Omega(n \log n)$, and the unnecessary $\log n$ gap arising from comparing to the cheapest proof vanishes. This is also consistent with the fact that any $O(n \log n)$ algorithm that ignores the sequence of keys and only treats them as a set is $O(1)$ instance optimal in the order-oblivious setting.\footnote{All algorithms that exploit certain presortedness in the input necessarily exploit the input sequence of keys. This corresponds to a smaller neighborhood than the $n!$ size automorphism neighborhood.}

Moving on to the case when $\vv{G}$ is not complete, we see that the above definition is not sufficient by the following observation. Consider the case when $\vv{G}$ is a complete bipartite graph, with all edges going in the same direction, i.e., from one partition $R$ to the other $B$. Now $|\N_{\A}(\vv{G})|=1$, and any algorithm that knows that it is operating on $\vv{G}$ has zero comparison cost\footnote{Recall that we do not charge to write down the transitive reduction which has size $O(|R||B|)$, only the query cost.}. However, any instance-unaware algorithm needs $\Omega(|R||B|)$ comparisons to verify that the instance is indeed $\vv{G}$. This suggests that we need a larger neighborhood than just $\N_{\A}(\vv{G})$. 

Let $\vv{E}$ denote the set of edges in the transitive reduction of $\vv{G}$, also called essential edges (Definition~\ref{def:transitiveReduction}). Define $\N_{E}(\vv{G})$ as the set of DAGs that differ from $\vv{G}$ in exactly one essential edge being flipped, and any other changes it may imply. Again, define the runtime of an algorithm on an instance $\vv{G}$ as its maximum comparison-cost on any instance in $\N_{E}(\vv{G})$, and define $OPT_{\A}(\vv{G})$ as the smallest comparison-cost of any algorithm. It is straightforward now to observe that if $\vv{G}$ is sortable, $OPT(S) \geq \text{cost}(\mathcal{H})$ and if $\vv{G}$ is not sortable, $OPT(S) \geq \text{cost}(TR(\vv{G}))$. \textit{Thus, we recover both definitions of competitive ratio by considering this one-flip neighborhood.} 

The algorithms we consider here are randomized, and we hence want a definition of instance optimality that allows for randomization.
The notion of ``instance optimal in the random-order setting'' of~\cite{afshani2017instance}, based on Yao's principle, compares implicitly to the optimal expected running time of a correct randomized algorithm (Las Vegas style). Our following definition does this directly:

\begin{definition}[$\alpha$-Instance Optimality]\label{introdef}
Let $\N(\vv{G})= \N_{E}(\vv{G}) \cup \N_{\A}(\vv{G})$. 
Define the runtime $T_{A}(\vv{G})$ of a randomized algorithm $A$ on an instance $\vv{G}$ as its expected comparison-cost on instance~$\vv{G}$. %
Let $\C(\vv{G})$ be the set of randomized (Las Vegas) algorithms that are correct for all instances in~$\N(\vv{G})$.
Define $\OPT(\vv{G}) = \inf_{A\in \C(\vv{G})} \max_{\vv{G}'\in \N(\vv{G}) } T_{A}(\vv{G}')$. For some $\alpha \geq 1$, an algorithm $A$ is called $\alpha$-instance optimal if for every instance $\vv{G}$, $T_{A}(\vv{G}) \leq \alpha \OPT(\vv{G})$.
\end{definition}

\subsection{Lower Bounds on OPT}\label{subsec:lower}

We first refine our notion of instance-optimality to make it more amenable to deriving lower bounds.

\begin{definition}[Instance Optimality Distribution]\label{instDist}
Let $\C'(\vv{G})$ be the set of deterministic algorithms that are correct for all instances in~$\N(\vv{G})$.
Let $\D(\vv{G})$ be the uniform distribution over $\N_{\A}(\vv{G})$.
Define $\OPT(\vv{G}) = \inf_{A\in \C'(\vv{G})} \mathbb{E}_{\vv{G}' \sim \D(\vv{G}) } T_{A}(\vv{G}')$.
\end{definition}

An application of Yao's principle \cite{yao77} shows that for any $\vv{G}$: $\OPT_D(\vv{G}) \le \OPT(\vv{G})$. Note that the optimal algorithm is allowed to depend on~$\I$. We remark that while the following lower bounds would be easy to prove for algorithms unaware of $\I$ using adversary arguments, we prove this for $OPT$ (Definition~\ref{instDist}), which requires some extra care.

\begin{restatable}[Transitive Reduction or Verification lower bound]{lemma}{lemverification}\label{lem:verification}
    Let~$\I$ be an instance of bipartite sorting, let $K\subset \vec E_{\I}$ be its transitive reduction, %
    and define~$C_V=|K|$.
    Then, any algorithm that is correct for all inputs from $\N(\I)$ must perform at least $C_V$ comparisons.
\end{restatable}

\begin{proof}
    Assume there exists an algorithm~$\A$ that is correct for all instances in $\N(\I)$ simultaneously that performs at most $C_V-1$ comparisons on input~$I$. 
    This means that there must exist an edge~$e$ on the transitive reduction that is not verified by~$\A$.
    As $\A$ is deterministic, it would report~$\I$ as output even when the input had~$e$ flipped because all other edges have the same direction as in~$\I$, as we will argue now:
    If this (non-flipped) edge is between two stripes of size one, the two endpoints merge into two other stripes, but no edges changes direction.
    If this edge is between two stripes and both of them have size at least 2, then we create one additional inversion without changing other directions.
    If one stripe has size one and the other has size at least two, one element is moved from the size-at-least-two stripe to a neighboring stripe.  
    Again, no other bichromatic edges (implied or not) change their direction. 
\end{proof}

\noindent While the above lower bound is natural as it is the cheapest proof, Definition~\ref{instDist} now allows for the following lower bound that captures the complexity of instances where transitive reduction is too weak a measure (recall the instance in Definition~\ref{def:one-flip}).

\begin{restatable}[Inversion finding lower bound]{lemma}{leminvLB}
    \label{lem:invLB}
    Let~$\I$ be an instance of bipartite sorting with $n\ge2$ red and $m\ge2$ blue elements, where not all comparisons come out the same, and define
    \[
        C_I= \frac{nm}{\min(|\{(r,b) \in R\times B\mid r<b\}|,|\{(r,b) \in R\times B\mid r>b\}|)}.
    \]
    Under the distribution of Definition~\ref{instDist}, any deterministic algorithm~$\A$ that does at most $C_I/2$ comparisons must fail with probability (at least) 1/8.
\end{restatable}

\begin{proof}
    Let~$\D$ be the uniform distribution over~$\N_A(\I)$. 
    Remember that $\N_A(\I)$ contains all instances where the stripes are internally arbitrarily permuted.
    Observe $|\N_A(\I)|\ge 4$ by the bounds on $n$ and $m$.
    W.l.o.g., assume $|\{(r,b) \in R\times B\mid r<b\}| < |\{(r,b) \in R\times B\mid b<r\}|$, i.e., 
    the usual outcome of a comparison is $b<r$ and the inversion is $r<b$.
    Let $p=1/C_I = |\{(r,b) \in R\times B\mid r<b\}|/nm \le 1/2$  be the probability that a randomly chosen pair of elements is an inversion. 
    By Yaos principle,
    let~$\A$ be a deterministic algorithm and think of it as a decision tree~$\T$ where nodes are red-blue comparisons and non-inversion go to the left, inversions go to the right.
    Each leaf of the tree is marked with an output (that declares which instance was represented by the input), or a failure output.
    
    \noindent\textbf{Claim:}
    Let~$v_k$ be the node~$v$ that is reached by~$k$ comparisons returning ``non-inversion'' (i.e. the leftmost node of~$\T$ at depth~$k$).
    When input is drawn from~$\D$, the node~$v_k$ is reached with probability at least~$1-kp$.
    
    \noindent\textit{Proof of Claim:} 
    An input reaches~$v_k$ if $k$ (potentially dependent) random experiments all came out as ``non-inversion'', each having a probability $1-p$. 
    The claim follows from a union bound over the fail events.\qed

    From the claim it follows that if $v_k$ is a leaf for $k\le C_I/2$, the algorithm must fail with probability at least 1/4: Then $kp\le C_I/2\cdot 1/C_I=1/2$ and $1-kp\ge 1/2$, so half of the inputs reach~$v_k$. 
    Because there are at least four inputs, at least two reach~$v_k$, but it can only be labeled with one, the other(s), which stand for at least 1/4 of the inputs in $\N_{A}(I)$, make the algorithm fail.
\end{proof}

Finally, we will need to combine various lower bounds from different subproblems. Let $\I$ be an instance and $ (S_1,\cdots,S_k)$ its stripes (see Figure~\ref{fig:stripes}). 
Consider pairs of indices $(a_1,b_1), \cdots, (a_{\ell},b_{\ell})$, where for all $1 \leq j \leq \ell$, $a_j$ and $b_j$ both belong to $\{1,\cdots,k\}$, and $a_j<b_j< a_{j+1}<b_{j+1}$ for all $j<\ell$.
For $1\leq j \leq \ell$ define the subinstance $\I_{j}$ by the subgraph of $\vec G_\I$ on the vertices~$V_j=\bigcup_{i=a_j}^{b_j} S_{i}$. 

\begin{restatable}[Decomposition into Lower Bounds for Subproblems]{lemma}{lemdecomposition}
\label{lem:decomposition}
    For $1 \leq j \leq \ell$, let $\I_j$ be a subinstances of $\I$ as above. 
    Then $\OPT_D(\I) \geq \sum_{j=1}^{\ell} \OPT_D(\I_{j})$.
\end{restatable}

\begin{proof}
    As we are working with Definition~\ref{instDist}, we first have to check that an algorithm that is correct on $\N(\I)$ is actually correct on each $\N(\I_j)$.
    To this end observe that every edge flip in $\I_j$ is also an edge flip in $\I$, and that any permutation of the labels/names in $\I_j$ is also a permutation in $\I$.

    Run the algorithm~$A$ on an instance~$\I'$ drawn from~$\D(\I)$, and let~$X$ be the random variable describing the number of comparisons of~$A$.
    Define the random variables~$X_j$ to be the number of comparisons between vertices in subproblem~$\I'_j$.
    Then $\sum X_j \le X$. 
    
    Let's conceptually draw $\I'$ from $\D(\I)$ by first finding a position in the input list (name) for all vertices not in $\I_j$ and finally draw names for the vertices in $\I_j$.
    Now we can think of~$A$ as deterministic algorithm for $\I_j$ by considering all comparisons not in~$\I_j$ as fixed, and we get $E[X_j]\ge \OPT_D(\I_j)$.
    The statement of the lemma follows by linearity of expectation.

\end{proof}

\subsection{InversionSort and its $O(\log^3 n)$ instance-optimal guarantee}\label{subsec:inversionsortanalyze}

A generic state of InversionSort will be defined using a \emph{backbone}, which is a sequence of elements of alternating colors, called \emph{representatives} or \emph{pivots}. 
Each representative will be assigned a \emph{bucket}, which is a set of elements of the same color that lie between the two neighboring representatives of the other color on the backbone.

InversionSort makes progress from one state to the next by performing three steps: a) finding an \emph{inversion} (which is defined soon)
between neighboring representatives on the backbone, b) inserting this inversion on the backbone, and c) pivoting with these two elements, thereby refining the buckets. 

\subsubsection{Description of InversionSort}
\begin{algorithm}
    \caption{Algorithm InversionSort}
    \label{alg:InversionSort}
    \begin{algorithmic}
        \REQUIRE elements $R$ red, $B$ blue
        \STATE create trivial backbone $\mathcal{B}$ from $R$ and $B$, see Definition~\ref{def:backbone}
        \STATE $\eta \leftarrow 0$
        \WHILE{there is an active subproblem (see Definition~\ref{def:active}) in $\mathcal{B}$}
        \STATE $\eta \leftarrow \eta+1$
        \FOR{each active (see Definition~\ref{def:active}) bucket $s$}
            \STATE Sample one element~$x_s$
        \ENDFOR
        \FOR{each active subproblem between buckets $s$ (left), and $q$ (right)}
            \STATE Test for inversion between $x_s$ and $x_q$
        \ENDFOR
        \FOR{each active subproblem $X_i,X_{i+1}$ where $\eta$ - mark (age) > $|X_i||X_{i+1}|$ }
        \STATE do all comparisons between $X_i$ and $X_{i+1}$
        \STATE update the backbone and certificates accordingly
        \STATE the subproblem is finished, i.e. no longer active
        \ENDFOR 
        \FOR{each found inversion}
        \STATE update the backbone, including splitting buckets and resampling pivots %
        \STATE mark new subproblems with round $\eta$ as age
        \ENDFOR
        \ENDWHILE
    \end{algorithmic}
\end{algorithm}

\begin{definition}\label{def:backbone}
[Backbone, Representatives, and Buckets]
The backbone consists of a totally ordered, alternating list of \emph{representatives}
$(u_0,u_1,u_2,\ldots,u_{2k},u_{2k+1})=(r_0,b_1,r_2,\ldots,r_{2k},b_{2k+1})$,
where $r_{2i} \in R$ and $b_{2j+1} \in B$ 
with $r_{i}<b_{i+1}$ and $b_i <r_{i+1}$.
Here, $r_0$ is an artificial red element that is smaller than all elements, and the last element $b_{2k+1}$ is an artificial blue element that is larger than all elements.
The representatives define the \emph{buckets}
$(X_0,X_1,X_2,\ldots,X_{2k},X_{2k+1})=(R_0,B_1,R_2,\ldots,R_{2k},B_{2k+1})$ by 
$R_i = \{ x\in R \mid b_{i-1} < x < b_{i+1} \} \setminus \{r_{i}\}$
and 
$B_i = \{ x\in B \mid r_{i-1} < x < r_{i+1} \} \setminus \{b_{i}\}$.
Here, as a convention, the representative is not included in the bucket. Again, 
$R_0 = \{ x\in R \mid x < b_{1} \}$
and
$B_{2k+1} = \{ x\in B \mid r_{2k} < x \}$
are special cases.
\end{definition}

\begin{definition}\label{def:active}
    [active subproblems and buckets]
    As long InversionSort did not create a certificate that there are no further inversions between $x_i$ and $x_{i+1}$, the \emph{subproblem} defined by the buckets $X_i$ and $X_{i+1}$ is called active, and so are $X_i$ and $X_{i+1}$.
\end{definition}

Our previous work \cite{bichromaticITCSarxiv} now defines an \emph{inversion}, which gives the algorithm its name. 
Consider adjacent representatives $u_i$ and $u_{i+1}$, their corresponding adjacent buckets~$X_i$ and $X_{i+1}$, and a bichromatic pair $(x,y)$ of elements $x\in X_i$ and $y\in X_{i+1}$.
Observe that $x$ and $y$ are not ordered by transitivity of the backbone.
Because $x$ and $y$ are of different color, they can be compared.
If $y<x$, the pair is called an \emph{inversion}. This allows one to extend the backbone: we get
$u_i<y<x<u_{i+1}$, which is a chain of actual comparisons between elements of alternating color. %

\noindent \textit{The only way to find an inversion in a bipartite setting (this is where the bichromatic setting is different) is to uniformly at random, from all pairs in $X_i$ and $X_{i+1}$, pick $x$ and $y$}.
If the fraction of inversions is $p$, then the probability of finding an inversion is~$p$ and the expected number of trials to find one is~$1/p$.

InversionSort starts by (trivially) having the backbone consist only of the artificial smallest red element~$r_0$ and largest blue element~$b_1$, and $R_0=R$ and $B_1=B$. For a given backbone $(u_0,u_1,u_2,\ldots,u_{2k},u_{2k+1})=(r_0,b_1,r_2,\ldots,r_{2k},b_{2k+1})$, InversionSort first, for each pair $X_{i},X_{i+1}$ of adjacent buckets that have not yet found an inversion or a 
proof that there is no further inversion (i.e., reached the adjacent-stripe verification bound), in a round-robin manner, does one round of inversion-searching by randomly comparing pairs of elements in adjacent buckets.
If this leads to an inversion, the inverted pair is saved and the algorithm moves to the next pair of adjacent buckets. At the end of the round, all identified inversions are considered and used to extend the backbone.
Then InversionSort splits existing buckets by pivoting with new elements on the backbone. 
Because there is at most one pair of inversions between each two neighboring representatives on the backbone, each element is compared to at most two new representatives in each round.

This reestablishes the backbone and creates some new pairs of neighboring buckets, for which InversionSort initializes the inversion finding procedures. The algorithm stops once all neighboring pairs of buckets are shown to not have an inversion, i.e., the comparisons in the verification bound between neighboring stripes have been performed.

\noindent\textbf{Analysis.} \cite{bichromaticITCSarxiv} visualizes a run of InversionSort as a ternary (refinement) tree, where nodes correspond to subproblems. 
For an internal node~$v$, there is a corresponding subinterval on the backbone defined by two consecutive pivots, say a blue pivot followed by a red pivot, $b_{v} < r_{v}$.  If InversionSort finds an inversion~$y<x$ ($x$ is blue and $y$ is red) between~$b_v$ and $r_v$, then~$v$ has three children with the respective pivots $(b_v,y)$, $(y,x)$, $(x,r_v)$. 

The random nature of the inversion searching of InversionSort, is made precise in Lemma~11 in \cite{bichromaticITCSarxiv}. However, a stronger version of this lemma applies for bipartite setting with the same proof.

\begin{lemma}[Randomness in Inversion Finding, Stronger Version of \protect{\cite[Lemma 11, p11]{bichromaticITCSarxiv}}]\label{lem:uniform}
    At any stage of the InversionSort, consider a successful inversion finding procedure, which finds an inversion $y<x$ between representatives $u_i<y<x<u_{i+1}$. 
    Say, w.l.o.g., that $u_i$ is red and $u_{i+1}$ is blue, and hence $x$ is red and $y$ is blue.
    \begin{enumerate}
        \item for any~$y\in X_{i+1}$, conditioned on $y$ being in the inversion, $x$ is uniformly distributed among all the red elements in $R_y = \{x \in X_i \mid y<x<u_{i+1}\}$, an. 
        \item for any~$x\in X_{i}$, conditioned on $x$ being in the inversion, $y$ is uniformly distributed in $B_x=\{y \in X_{i+1} \mid u_i<y<x\}$.
    \end{enumerate}
    \end{lemma}

This gives a bound on the height of the tree.

\begin{lemma}[Height of the refinement tree \protect{\cite[Theorem 5, p11]{bichromaticITCSarxiv}}]
    \label{thm:height}
    Let $\T$ be the refinement tree of running InversionSort on an instance~$\I$ with $N=n+m$ elements.
    With high probability in~$N$, the height of~$\T$ is $O(\log N )$.
\end{lemma}

\noindent\textit{Handling Overlaps.} Because of the overlapping nature of the problem, InversionSort cannot easily focus on elements between neighboring representatives.
For example, for the child indicated by pivots $(y,x)$, instead of only getting the reds and blues that actually lie in this range as input, InversionSort instead has to also work with the red elements contained in $(b_v,y)$ and the blue elements inside $(x,r_v)$. This ``spill-in'' from the neighboring subintervals on the backbone needs to be analyzed.

As is argued in~\cite{bichromaticITCSarxiv}, the cost of bichromatic inversion search procedure of InversionSort is justified by the subinstance (part of the Hamiltonian) between the neighboring elements on the backbone. 
In Section~\ref{sec:instOptInvSort}, we will analyse this in the bipartite setting using the notion of instance optimality.
However, if the spill-in for this subproblem is too large, inversion search is too costly.
Hence, \cite{bichromaticITCSarxiv} identify subproblems that do not have too much spill-in from their neighbors, and call these subproblems \emph{unaffected}. 
Inversion search in unaffected subproblems can be charged to their subinstance. 
More precisely, \cite{bichromaticITCSarxiv} show that 
at any time, with high probability, at least roughly a $1/(\log N)^2$ fraction of all current problems are unaffected. Accounting over the whole tree, including the pivoting, introduces another $\log N$ factor corresponding to the depth of the tree. 

\subsubsection{Putting everything together: Proof of Theorem~\ref{thm:bipartiteinversionsort}}\label{sec:instOptInvSort}

We now complete the proof of our main result for bipartite sorting, restated here for convenience.

\bipartitetheorem*

We will show how to charge the comparisons performed by InversionSort to the lower bounds presented in Section~\ref{subsec:lower}.
First, Lemma~\ref{lem:uniform} and Lemma~\ref{thm:height} imply that the refinement tree height $h=O(\log n)$ and hence the pivoting cost is~$O(n\log n)$.
Second, the considerations about unaffected subproblems remains valid, %
there are at most $O(\log^2 n)$ affected subproblems per unaffected subproblem.
The inversion searching cost in each unaffected subproblem is justified by Lemma~\ref{lem:invLB} if an inversion is found, otherwise by Lemma~\ref{lem:verification}.
Adding the extra $O(\log \mathcal{N})$ factor (height of the tree), the next lemma completes the proof of Theorem~\ref{thm:bipartiteinversionsort}.

\begin{lemma}
    \label{lem:treeAccounting}
    Let $\T$ be the refinement tree of height~$h$ for a run of InversionSort on instance~$\I$, and let~$V_\T$ be the set of nodes of~$\T$.
    Then
    \[
    \sum_{v\in V_\T} \OPT\left( \I_v \right) \le 2h\OPT(\I)
    \]
\end{lemma}
\begin{proof}
    It suffices to show that for each level $\L$ of~$\T$ the inequality
    \( \sum_{v\in \L} \OPT\left( \I_v \right) \le 2 \OPT(I) \) holds.
    Note that the subinstances $\I_v$ of the same polarity in any~$\L$ do not share vertices.
    Hence the decomposition Lemma~\ref{lem:decomposition} is applicable and the lemma follows because there are two polarities.
\end{proof}

\section{Result on Sorting with Priced Information: Lower bound does not extend}\label{sec:generalsorting}

Recall that there exists an instance demonstrating the lower bound of $\Omega(n)$ on the competitive ratio of any algorithm that finds the \textbf{maximum} of a set of $n$ elements. Announced in the original paper by Charikar et al.\cite{charikar2002query}, this was spelled out one year later by Gupta and Kumar~\cite{gupta2001sorting}]. 
The instance is simple to describe (see Appendix~\ref{apx:maximum_lb}), and all comparisons in it have costs in $\{0,1,n,\infty\}$. Although this example was never formally stated for sorting, its discovery seems to have dampened efforts to study (either better algorithms, or lower bounds for) the general version of sorting with priced information in the past 20 years. 

In this section, we prove Theorem~\ref{mainresult1}, that shows that the $\Omega(n)$ lower bound for maximum with costs in $\{0,1,n,\infty\}$ cannot extend to sorting. 

\resultongeneralsorting*

While counterintuitive at a first glance (after all, the cost to sort is at least the cost to find the maximum),  the simple explanation is that the cheapest proof for sorting is also more expensive than that of the maximum. This opens up the problem of arbitrary-cost sorting once again - is there a $\Omega(n)$ lower bound for sorting with arbitrary costs, or can our $\tilde{O}(n^{3/4})$ algorithm be extended to an $o(n)$ competitive algorithm for arbitrary costs?

Theorem~\ref{mainresult1} is achieved by \emph{first developing an algorithm for DAG sorting}. Why consider the case of an unsortable DAG, when the DAG we have is sortable? Here is the reasoning. If we consider greedy algorithms for sorting with priced information, it is natural to try to discover as much of $\vv{G}$ as possible with low-cost edges\footnote{We remark that some time after we first uploaded a version of Theorem~\ref{mainresult1} to Arxiv, a preprint by [Jiang, Wang, Zhang and Zhang, Arxiv \cite{jiang2023algorithms}] was uploaded, where the authors also use an algorithm for DAG sorting (they call it GPSC) parameterized by $w$ and extend our $\widetilde{O}(n^{3/4})$ algorithm to obtain a $\widetilde{O}(n^{1-1/2W})$ competitive algorithm for sorting with priced comparisons with at most $W$ distinct costs. For our setting when $W=4$ they re-derive our result separately as their main theorem  would give a $n^{7/8}$-competitive ratio.}. However, note that the sub-DAG $\vv{G}_{\leq w}$ consisting of edges with cost at most $w$ in $\vv{G}$ may not be sortable, which is exactly the problem of DAG sorting.

We set up some notation first. Given an undirected complete graph $G$ with costs in $\{0,1,F,\infty\}$, let $\vv G$ denote the underlying DAG that contains a directed Hamiltonian path. Define ${\vv G}_{0}$ as the DAG obtained by revealing all cost $0$ edges, observe that it may not have a Hamiltonian path, and let $w_0$ be the width of ${\vv G}_{0}$. Similarly, denote by ${\vv G}_{01}$ the DAG obtained by revealing all cost $0$ and $1$ edges; ${\vv G}_{01}$ may not have a Hamiltonian path either, and let $w_{01}$ be the width of ${\vv G}_{01}$. Finally, let $k_{1}$ and $k_{F}$ be the number of cost $1$ and cost $F$ edges on the Hamiltonian path in $\vv G$.

\subsection{Algorithm details}

The following is our proposed algorithm for the $0,1,F,\infty$ cost version of sorting with priced information. Below, we will abbreviate the algorithm by Huang, Kannan and Khanna~\cite{6108244} for the $1,\infty$ setting, by HKK.

\begin{algorithm}
    \caption{Algorithm for $0,1,F,\infty$ cost}
    \label{alg:competitiveAlg}
    \begin{algorithmic}
        \REQUIRE undirected graph $G=(V,E,c)$ with costs $c(e) \in \{0,1,F,\infty\}$
        \ENSURE the total order (directed Hamiltonian)
        \STATE Probe all cost zero edges
        \STATE Run the following 4 algorithms in parallel, performing one comparison from each. If any of them discover the Hamiltonian, report the edges in the Hamiltonian path, and abort the other algorithms
        \STATE \qquad $\circ$ Set $F=\infty$. Run the HKK algorithm on the cost 1 edges, starting from $\vv{G_0}$ 
        \STATE \qquad $\circ$ Set $F=1$. Run the HKK algorithm on the cost $1$ edges, starting from $\vv{G_0}$ 
        \STATE \qquad $\circ$ Run Algorithm~\ref{alg:hamPred} on the cost 1 edges, starting from $\vv{G_0}$ 
        \STATE \qquad $\circ$ Find the 0-1 DAG using Theorem~\ref{get01dag}, use Algorithm~\ref{alg:hamPred} with cost $F$ edges on it.
    \end{algorithmic}
\end{algorithm}

The running time of the final algorithm will then be a minimum of four running times. We briefly explain the first two algorithms, and then explain in detail the last two. Recall that HKK algorithm runs when costs are $1$ or $\infty$. On an input with costs in $\{0,1,F,\infty\}$, the first algorithm pretends that cost $F$ edges are forbidden too, i.e., $F=\infty$, and probes whatever edges HKK would have probed from the cost 1 edges. Clearly, this will not find the Hamiltonian if it contains cost $F$ edges, as they aren't queried. However, in the case that the Hamiltonian does not contain cost $F$ edges, it will sort the input, and stops. The second algorithm does the opposite: it does not differentiate between cost $F$ and cost $1$ edges, and probes them if they would have been probed by the HKK algorithm. If run for long enough, this will find the Hamiltonian, and is stopped once it does so.

\subsubsection{Algorithm~\ref{alg:hamPred} : Hamiltonian By Predecessors}

The third algorithm in Algorithm~\ref{alg:competitiveAlg} is Algorithm~\ref{alg:hamPred}, which is also used as the second half of the fourth algorithm in Algorithm~\ref{alg:competitiveAlg}. This algorithm finds a Hamiltonian path in a partially revealed DAG. It utilizes Lemma~\ref{lem:hamiltonian_predecessor}, that generalizes binary search to searching for predecessors of a vertex in a DAG of width~$w$.  For two DAGs $D'$ and $D$ on the same set of vertices, we will write $D' \subset D$ if all the directed edges in $D'$ are also contained in $D$. 

\begin{restatable}[Hamiltonian by predecessor search]{lemma}{hamiltonianpredecessor}\label{lem:hamiltonian_predecessor}
    Let $D' \subset D$ be two DAGs on the vertex set $V$ and assume that $D$  contains a Hamiltonian path. 
    Assume that the Hamiltonian path in~$D$ contains a set $\vv S$ of~$k$ edges that are not in $D'$, and let $S$ be the undirected version of $\vv S$. Let $E$ be a set of edges that can be queried and assume $S \subset E$.
    Let~$w$ be the width of~$D'$. Then, $k+1 \ge w$ and
    the Hamiltonian in $D$ can be found with $O(wk\log n)$ queries on edges in $E$.
\end{restatable}

Below is the pseudocode for Algorithm~\ref{alg:hamPred}. It uses in turn a predecessor searching subroutine that is captured by the following simple lemma. Lemma~\ref{lem:hamiltonian_predecessor}.

\begin{lemma}[Predecessor search in DAG]\label{lem:multi_pred_search}
    Given a DAG~$D'=(V,E')$ of width~$w$, and a vertex~$v\in V$, $|V|=n$, let $D$ be the DAG obtained by extending $D'$ by probing all edges involving $v$. Define $P_v= \{ u\mid (u,v) \hbox{ is in the transitive reduction of } D \}$. There exists an algorithm  that computes $P_v$ with $O(w\log n)$ queries, and runs in $O(n^2)$ time.
\end{lemma}
\begin{proof}
    Observe that any chain in~$D'$ can contain at most one element of~$P_v$, and hence $|P_v| \le w$.
    Indeed, we can run one binary search on each of the $w$ chains in $D'$, leading to at most~$w$ candidate predecessors.
    The number of queries is easily seen to be $O(w\log n)$ after computing an optimal partitioning into chains in polynomial time.
\end{proof}

\vspace{2mm}\noindent\textbf{Proof of Lemma~\ref{lem:hamiltonian_predecessor}} 
    By Dilworth's Theorem, $D'$ can be partitioned into~$w$ chains.
    To show $k+1 \ge w$, assume otherwise, $w > k+1$, and let~$A$ be $k+2$ non-comparable vertices in~$D'$.
    By the pigeonhole principle, there must be two vertices of~$A$ in the same of the $k+1$ stretches of cost 0 edges on the Hamiltonian, a contradiction to them being incomparable.

    To prove that Algorithm~\ref{alg:hamPred} performs at most $O(kw\log n)$ queries,
    observe that adding edges to~$D'$ does not increase its width.
    In the while loop, as long as~$D'$ is not the Hamiltonian path, let~$S$ be the first layer of a BFS traversal of the transitive reduction of~$D'$ with $|S|\ge 2$, and observe that $S$ is an antichain and hence $|S|\le w$.
    All vertices of $S$ but one have their incoming edge on the directed Hamiltonian not yet revealed / queried: 
    Assume there are two vertices $b \neq d \in S$ and their predecessors $a\prec b$ and $c\prec d$ are both already in $D$.
    Then, $b$ and $d$ are not sources in $D$, and hence $S$ must be the set of successors of a vertex~$v$.
    Additionally, there is only a single source $s$ in~$D$, and the set $\{x\mid x < v\}$ forms a chain in the transitive reduction of~$D$ starting in~$s$.
    This means $w\le a < b$ and $w\le c < d$ contradicting them being different.

    By the above arguments, $D$ contains the Hamiltonian with only $k-|S|+1$ unrevealed edges missing.
    We used $O(|S|w\log n)$ queries to reduce the number of unrevealed edges by $|S|-1$ for $|S|\ge 2$, hence each search creating a missing edge of the Hamiltonian, and this search must justify at most one additional such search.
    Hence, the total number of queries to arrive at the Hamiltonian is~$O(wk\log n)$. 

Binary searching for a vertex $v$ into one of the $w$ chains takes $O(\log n)$ probes, and in $O(w \log n)$ probes one is sure to have at least discovered one edge from the Hamiltonian, namely the incoming edge to $v$. This can then be repeated $k$ times, revealing the Hamiltonian.
    
    \begin{algorithm}
        \caption{Hamiltonian By Predecessors}
        \label{alg:hamPred}
        \begin{algorithmic}
            \REQUIRE undirected $G=(V,E)$ defining which comparisons are allowed
            \REQUIRE DAG $D$ of already probed edges (initially the cost 0 edges)
            \ENSURE The updated $D$ contains a directed Hamiltonian
            \STATE create (and maintain) a chain decomposition $C$ of the transitive reduction of $D$
            \WHILE{$D$ has width $w$ >1, i.e. is not the intended result}
                \IF{The transitive reduction of $D$ has several sources}
                    \STATE Let $S$ be the set of these sources
                \ELSE
                    \STATE Let $v$ be the lowest vertex with more than 1 successors
                    \STATE Let $S$ be the set of successors of~$v$
                \ENDIF
                \STATE 
                \FOR{each $u$ in $S$}
                    \STATE Find all predecessors of $u$ in $D$ (Lemma~\ref{lem:multi_pred_search}), adding answers to $D$ 
                    \STATE
                    {\hfill $\backslash\backslash$ there are at most $w$ such predecessors}
                \ENDFOR
            \ENDWHILE
        \end{algorithmic}
    \end{algorithm}

\subsubsection{The fourth algorithm in Algorithm~\ref{alg:competitiveAlg}}

We will develop another algorithm, that proceeds in two steps: a) compute only the 0-1 DAG, ${\vv G}_{01}$, and b) find the cost $F$ edges ($k_F$-many of them) on the Hamiltonian path. Step b) is performed using Algorithm~\ref{alg:hamPred}. If $k_F=0$, step a) recovers the complete Hamiltonian path. Before we state the DAG sorting algorithm for step a), we note that it only needs to output the transitive reduction of ${\vv G}_{01}$.

\begin{restatable}{theorem}{generalizednthreehalf}\label{get01dag}
    There is a poly-time randomized algorithm that w.h.p. solves DAG sorting for an instance $\vec{G}$ with edge costs in $\{0,1,\infty\}$, using $O(\min(wn^{3/2}\log n,n^2))$ comparisons, where~$w$ is the width of~$\vec{G}$.
\end{restatable}

We defer the proof of Theorem~\ref{get01dag} for now and analyze our algorithm assuming it.

\subsection{Analysis of Algorithm~\ref{alg:competitiveAlg}}

\begin{lemma}
    
Algorithm~\ref{alg:competitiveAlg} incurs the following costs
\[ \begin{cases} 
      O(\min(n^{1.5}\log n \,,\, w_0k_1\log n)) & \text{if } k_F = 0 \\
      O(\min(Fn^{1.5}\log n\,,\,w_{01}n^{1.5}\log n+Fw_{01}k_F\log n)) & \text{if } k_F > 0
   \end{cases}
\]
\end{lemma}

\begin{proof}
    If $k_F=0$, the first algorithm that ignores cost $F$ edges by setting $F=\infty$ never probes a cost $F$ edge, and finishes in $O(n^{1.5}\log n)$ comparisons (this is the cost of the algorithm by Huang, Kannan and Khanna~\cite{6108244}). Since the DAG formed by cost 0 edges has width $w_0$ and $k_F=0$, $w_{0} \leq k_1+1$. The third step running  Algorithm~\ref{alg:hamPred} finishes after at most $O(w_0k_1 \log n)$ comparisons, by Lemma~\ref{lem:hamiltonian_predecessor}.

    If $k_F >0$, the first term comes from running HKK after setting $F=1$: the true cost of probing an edge is at most a factor $F$ larger. Finally, step 4 of Algorithm~\ref{alg:competitiveAlg} runs the algorithm in Theorem~\ref{get01dag} first, incurring at most $w_{01}n^{1.5}\log n$ many comparisons. With the 0-1 DAG obtained using Theorem~\ref{get01dag}, Algorithm~\ref{alg:hamPred} now inserts at most $k_F$ many edges in the Hamiltonian, probing at most $w_{01}\log n$ many edges for each. Every probe costs $F$, for a total of $w_{01}n^{1.5}\log n+Fw_{01}k_F\log n$.
\end{proof}

\subsubsection{Proof of $\tilde{O}(n^{3/4})$ competitive ratio of Algorithm~\ref{alg:competitiveAlg}}

We claim that the competitive ratio is always bounded by $O(n^{3/4}\log n)$. Observe that the cost of the Hamiltonian is $k_1+Fk_F$. 
If $k_1=k_F=0$, the Hamiltonian has a cost of 0 and our algorithm finds it at cost 0. 
From now on, we assume not both of $k_1$ and $k_F$ are $0$.

Consider the case $k_F=0$ first. Note that this implies that the width $w_{01}$ of $\vv G_{01}$ is $1$. First consider the subcase when $w_0 \leq n^{3/4}$. 
In this case, the competitive ratio is bounded above by $O(w_0 k_1 \log n) / k_1 = O(w_0 \log n) \leq O(n^{3/4}\log n)$. In the subcase when $w_0 > n^{3/4}$, observe that this implies that $k_1 \geq n^{3/4}$ which implies that the competitive ratio is bounded above by $O(n^{1.5}\log n)/k_1 \leq O(n^{3/4}\log n)$.

Next, consider the case $k_F \geq 1$, and the cost of the Hamiltonian is at least $Fk_F$. Since $w_{01} \leq k_F + 1$, the cost of the algorithm is at most $O(w_{01} n^{1.5}\log n + Fw_{01}k_F\log n) < O(k_F n^{1.5}\log n + Fw_{01}k_F\log n)$, and dividing by $FK_F$ (the lower bound on the cost of the Hamiltonian), we get  a competitive ratio of at most $O((n^{1.5}/F+ w_{01}) \log n)$. Since $F \geq n^{3/4}$, this ratio is $O(n^{3/4}\log n)$ as long as $w_{01} \leq n^{3/4}$. Else if $w_{01} > n^{3/4}$, we observe that $k_F \geq n^{3/4}$, and then the $Fn^{1.5}\log n$ query cost gives us a competitive ratio of at most $Fn^{1.5}\log n / F k_F \leq n^{3/4}\log n$. Thus Theorem~\ref{mainresult1} is proved.\qed

\noindent It remains to prove Theorem~\ref{get01dag}, which is the topic of the next subsection.

\subsection{Proof of Theorem~\ref{get01dag}}

First, observe that if the width of~$\vec G$ is at least~$\sqrt{n}/4$, the statement of Theorem~\ref{get01dag} is easy to achieve by probing all edges.
Hence, let us assume the width is at most~$\sqrt{n}/4$. We will show that there is an algorithm that computes $\vv G_{01}$ with cost at most $O(w_{01}n^{1.5}\log n)$. 
This algorithm will only probe cost $0$ and $1$ edges, and  will be a generalization of the algorithm in \cite{6108244}. Note that while the algorithm in \cite{6108244} works on a cost $\{1,\infty\}$ setting under the promise of a Hamiltonian path in the true graph, our algorithm finds the transitive reduction of the DAG $\vv G_{01}$ of width $w_{01}$.

\noindent  Now we address the challenges posed in extending the work by Huang Kannan and Khanna \cite{6108244}.

\vspace{2mm} \noindent\textbf{Challenges in extending the results of Huang, Kannan and Khanna \cite{6108244}}: At a high level, the algorithm in \cite{6108244} alternates between three ways of making progress: 

\noindent\textbf{1.} Finding and probing balanced edges, defined as those that reduce the number of possible linear extensions of the current DAG by a $1-(1/ (e \sqrt{n})))$ factor. Finding such edges requires approximating the average rank of vertices under all possible linear extensions at all stages of the algorithm.

\noindent\textbf{2.} After estimating the indegree of vertices upto an additive error of $\widetilde{O}(\sqrt{n})$ by an $\widetilde{O}(n^{1.5})$ sampling procedure, the algorithm probes free edges, defined as the set of edges $(u,v)$ where the average rank of $u$ is smaller than the average rank of $v$, and $v$ has most $\widetilde{O}(\sqrt{n})$ unprobed incoming edges. Free edges that are balanced again reduce the number of linear extensions by a constant factor. Otherwise, they can contribute at most $\widetilde{O}(n^{1.5})$ to the total cost.

\noindent\textbf{3.} Binary Search - When there are no free edges, there must exist a set of $\sqrt{n}$ vertices with known total order (Lemma 3.5 in \cite{6108244}). The other vertices can perform binary search into these $\sqrt{n}$ vertices at a cost of $O(n \log n)$, and doing so removes these $\sqrt{n}$ vertices from the picture. The total cost of binary search is therefore $\widetilde{O}(n^{1.5})$. 

\textit{The third step of the algorithm is the step that guarantees a reduction in the problem size. However,the third step of this algorithm no longer works for DAG sorting}: the existence of  a set of $\sqrt{n}$ vertices with known total order crucially relies upon the existence of the Hamiltonian path. 

\vspace{2mm}\noindent\textbf{Proof of Theorem~\ref{get01dag}} All of the definitions, algorithms, and accounting to estimate the in-degree of a vertex in \cite{6108244} remain valid and unchanged. Observe that any topological sorting of the underlying directed graph, together with the undirected graph, reveal the directed graph. 
Define the \textbf{average rank of a vertex} as the average rank over all linear extensions of the true underlying directed graph. 
The following result implies that the average rank~$r$ on a path in (the transitive reduction of) a DAG is increasing by at least one per edge.
\begin{lemma}\label{lem:avg_plus_one}
    Let~$D=(V,E)$ be a DAG and~$r\colon V\to \Qpos$ be the average rank.
    Then for $(u,v)\in E$ we have $r(u)+1\le r(v)$.
\end{lemma}
\begin{proof}
    Let $\Pi$ be the set of all linear extensions that are compatible with~$D$, such that~$r(x)|\Pi| = \sum_{\pi \in\Pi} \pi(x)$.
    Then
    $|\Pi|(r(v)-r(u)) = \sum_{\pi \in\Pi} \pi(v)-\pi(u) \ge |\Pi|\cdot 1$.
\end{proof}

\begin{definition}[Convex vertex subset] 
In a DAG $G = (V,E)$, a subset of vertices $S \subseteq V$  is \emph{convex} if for every pair of vertices $u, v \in S$, every vertex $w$ on any directed path from $u$ to $v$ in $G$ is also in $S$.\footnote{Our definition of convexity differs from the definition in the metric graph theory (defined on undirected graphs), where convex subgraph contains the vertices of only the shortest paths between every pair of vertices~\cite{metric-graph-survey}.} %
\end{definition}

Hence, considering a subset of the vertices by an upper and a lower bound on the average rank, leads to a convex subset. Next, a vertex is \textbf{live} if there is an unprobed edge incident to it, otherwise it is \textbf{exhausted}. The \textbf{assumed graph} is the same directed graph as in \cite{6108244}.
An \textbf{active vertex} is one that has at most $4\sqrt n \log n$ unprobed in-edges in the assumed graph. A \textbf{free edge} is an unprobed edge~$(u,v)$ where the endpoint~$v$ is active. The proof of the next lemma is identical to that in \cite{6108244}.

\begin{lemma}[Generalization of Lemma~3.5 in \cite{6108244}]\label{lem:their35}
    The $\sqrt n$ live vertices with smallest average rank are all active.
\end{lemma}

\begin{restatable}{lemma}{theirthreesix}[Generalization of Lemma~3.6 in \cite{6108244}]\label{lem:their36}
    If there are no free edges, and the width of the underlying~$\vec G$ is at most~$\sqrt{n}/4$, then there exists a set $S$ of at least $(3/4)\sqrt n$ live vertices with known partial order who form a DAG of width at most~$\sqrt{n}/4$.
\end{restatable}

\begin{proof}
    Consider the set~$S$ of at most~$\sqrt{n}$ live vertices with smallest average rank.
    More precisely, we chose the largest upper bound on the average rank such that $|S|\le\sqrt{n}$.
    By Lemma~\ref{lem:avg_plus_one}, at most~$\sqrt{n}/4$ vertices can have the same average rank, such that $|S|\ge 3/4\sqrt{n}$.
    By Lemma~\ref{lem:their35}, all vertices of~$S$ are active.
    Let~$u,v\in S$ be a pair of vertices that have a directed path $P$ from~$u$ to~$v$ in $\vec G$.
    Then, all of this path~$P$ is in $S$, and all live vertices of~$P$ are in~$S$.
    Hence, because there are no free edges, and all vertices of~$P$ not in~$S$ are exhausted, all edges of~$P$ must be probed. 
    Hence, $S$ is convex.
    The statement on the width follows from a chain decomposition of~$\vec G$ remaining a chain decomposition for a convex subset of vertices.
\end{proof}

Note that Lemma~\ref{lem:their36} does not imply that the algorithm can, or should, identify precisely this set~$S$ defined in the proof.
Hence, the algorithm is going to approximate the smallest width subset of at least~$3/4 \sqrt{n}$ vertices among the live vertices.
More precisely, starting from an empty~$S$, it is going to iteratively find the longest (outside~$S$) chain among the live vertices (also using edges that are implied by transitivity via edges in~$S$).
It stops once~$S$ contains at least $3/4 \sqrt{n}$ vertices, and uses it as the DAG of small width in the setting of Lemma~\ref{lem:multi_pred_search}, and determine for every remaining live vertex its predecessors in~$S$, with~$O(w\log n)$ queries each, compared to the~$O(\log n)$ queries if there is a Hamiltonian.
Hence, the total number of queries increases from $O(n^{3/2}\log n)$ to $O(wn^{3/2}\log n)$, as claimed in Theorem~\ref{get01dag}, completing the proof.

\clearpage

\appendix
\section{Bipartite Sorting with Hamiltonian}
\subsection{Backbonesort: Quicksort Adaptation for Perfectly Interleaved Instance}\label{sec:backbonesort}

Consider the perfectly interleaved instance with $n=m$: the smallest element is a red, followed by a blue, followed by a red, and so on. 

\begin{definition}[Backbonesort]
Backbonesort has the same notion of a backbone as InversionSort.
It runs in rounds that alternates between pivoting with red and pivoting with blue.
In a round of pivoting with red (pivoting with blue is completely symmetrical), for each bucket of red elements, it chooses uniformly at random a red pivot.
Now, each blue bucket is (attempted to be) split with the two red pivots of the neighboring red buckets, splitting it into 1 (not splitting), 2 or 3 blue buckets.
For the new blue buckets that don't have a pivot/representative on the backbone yet, choose one uniformly at random, and put them on the backbone.
Finally, split all red buckets for which this makes sense using the newly chosen blue pivots.
\end{definition}

\begin{theorem}\label{backboneperfect}
With probability at least $1-1/n$, Backbonesort runs in time $O(n \log n)$ on the perfectly interleaved instance.
\end{theorem}
\begin{proof}
Consider an arbitrary red element~$r$. 
Consider one round of red pivoting, and let $B$ be the bucket of~$r$ before the round, and~$B'$ after the round.
Then, if $|B|>1$, with probability~1/8, $|B'|\le \frac34 |B|$:
In the natural order of the elements, with probability 1/4, the randomly chosen pivot is in the middle, i.e. rank 1/4 to 3/4, of the bucket, and not in the same half as the current pivot/representative of the bucket.
Now, the blue pivot is chosen uniformly between the old and the new pivot, and the probability of being in the middle half is at least 1/2 (the old pivot on the boundary, the new almost in the middle -- everything else has an even higher probability, potentially even~1).
Hence, with overall probability 1/8, $B$ is split in the middle half, and $|B'|\le \frac34 |B|$.

From the claim follows from a proof similar to that of Theorem~\ref{thm:height} that Backbonesort performs $O(\log n)$ rounds.
Because each round performs~$O(n)$ comparisons, the theorem follows. 
\end{proof}

\subsection{InversionSort for bipartite sorting with Hamiltonian}\label{sec:InvsortOnInterleaved}
While Theorem~\ref{thm:bipartiteinversionsort} shows only that InversionSort performs $O(N\log^3 N)$ comparisons on the perfectly interleaved instance, it is actually optimal.

\begin{theorem}
[InversionSort is optimal on perfectly interleaved]
\label{thm:InvsortOnInterleaved}
With probability at least $1-1/n$, InversionSort runs in time $O(n \log n)$ on the perfectly interleaved instance.
\end{theorem}
\begin{proof}
\textbf{(Sketch)}
From Theorem~\ref{thm:height} follows that the pivoting cost is as claimed.
To bound the inversion cost, compute for each backbone (in retrospect / knowing the instance) the quantity~$Q$ as the sum of the squares of the sizes of the subproblems.
If we can show that, with constant probability, $Q$ is reduced by a constant factor in each step of the algorithm, the Theorem follows.
Actually, it would be enough to show that each summand of~$Q$ is reduced by a constant factor with constant probability. 
While this is true for the unaffected subproblems, it can fail for affected subproblems, i.e., subproblems that have a significant spill-in.
Hence, let's bundle together, for an unaffected subproblem~$p$, all subproblems that it spills in, and all subproblems these spill in and so on. 
These are at two chains starting at~$p$, one to the left and one to the  right. 
By the definition of spill-in, and because the instance is the perfectly interleaved one, the sizes of the spilled-in problems decrease by a factor of 4, and the squares of the sizes by a factor of 16.  
Hence the contribution of~$p$ to~$Q$ is dominating the contribution of the bundle, and its probability to get substantially reduced is enough to show that the whole bundle has a constant probability to reduce its contribution to~$Q$ by a constant factor.
\end{proof}

The following theorem shows that Backbonesort is quite far from being instance optimal, in contrast to InversionSort.

\begin{theorem}\label{thm:balloon}
For every $n$, there exists a family of instances $I_{n}$ with $n$ red keys and $n$ blue keys, such that the expected number of comparisons performed by Backbonesort is $\Omega(n^{1.5})$, whereas the expected cost of $OPT(I_{n}) = O( n \log n)$.
\end{theorem}

\begin{proof}
Consider the instance $I$ given by the size vector $(n-\sqrt{n},1,1,\cdots,1,1,n-\sqrt{n})$. The first stripe of reds is size $n-\sqrt{n}$, followed by $\sqrt{n}$ perfectly interleaved reds and blues, followed by a blue stripe of size $n-\sqrt{n}$. We call the big stripes at either ends ``balloons'', to distinguish them from the singleton stripes in the middle.

\noindent\textbf{Analyzing Backbonesort on I:} 
Backbonesort will see as input two sets $R$ and $B$ of size $n$. First, it will pick a random pivot $r$ in $R$, and pivot the $B$ using this. 
With probability $1-1/\sqrt{n}$, this pivot will be from the red balloon, and with probability $1/\sqrt{n}$, it will be one of the reds in the singleton stripe. 
It is easy to see that until Backbonesort ends up picking a red or blue pivot not in the respective balloon, all its comparisons come out as red < blue. 

Thus Backbonesort will select a pivot from the balloons $\Omega(\sqrt{n})$ times in expectation before it selects a useful red element that is not in the balloon. However, since pivoting costs $O(n)$ comparisons every time, at this point Backbonesort has already performed $\Omega(n^{1.5})$ comparisons in expectation.

\paragraph{Analyzing OPT on $I$:} We upper bound OPT by analyzing the running time of InversionSort on this instance. Assume that  the first comparison of a red with a blue gives $r < b$, a likely scenario since these elements are more likely to come from the balloon than the singleton stripes. Inversion sort will create a backbone with a red pivot $r$ and a blue pivot $b$, and will try to find an inversion between them. 
Now there are $\Theta(n)$ inversions between the $\sqrt{n}$ reds and the $\sqrt{n}$ blues in the middle, so the probability of finding an inversion is $\Theta(n/n^2) = \Theta(1/n)$. 
Thus in expected time $O(n)$ InversionSort will find an inversion $y<x$, where $y$ is blue and $x$ is red, to have $r<y<x<b$ on the backbone.

At this point, the balloon instances on the two sides get separated, but they each spill-in into the middle subproblem, and thus the middle subproblem will be affected. It can be shown that with constant probability the middle subproblem will not affect the two balloon subproblems. 

Assume $y$ is at position $k_{1}$ from the red balloom and $x$ is at position $k_{2}$ from the red balloon. Then the expected time to find an inversion in the left subproblem is $(n-\sqrt{n}+k_{1})k_{1}/(k_{1})^2 \approx n/k_{1}$. Similarly the expected time to find an inversion in the right subproblem is $n/k_2$. For the middle subproblem, its slightly worse: it will be $n^2/(\sqrt{n}-k_1-k_2)^2$, which could be as large as $n$.

When the side balloon get refined further, only their neighbors will get their spill-in, and the middle problems at the core start to become unaffected. At the extreme case, the balloons on the size only have $O(1)$-sized interleaving in the subproblem attached to them, and their expected cost to find an inversion becomes $O(n)$. 

However, at any point in time, if one considers the work done searching for an inversion\footnote{As we showed before,
the pivoting cost is always $O(n \log n)$.} 
on a layer of the decomposition tree, then the work on any subproblem does not exceed $O(n)$ in expectation by the above argument.  Using our general lemma on the depth of the tree being $O(\log n)$ (Lemma~\ref{thm:height}) for any run of InversionSort, we get that the expected cost of InversionSort on $I$ is $O(n \log n)$. 
Furthermore, it is easy to see that the cost of InversionSort on all instances in $\N_{AE}(I)$ is also $O(n \log n)$ in expectation; these instances have roughly the same number of singleton stripes, and roughly the same sized-balloons. Thus we get that $OPT(I) = O(n \log n)$ in expectation, proving the theorem.
\end{proof}

\section{Lower bound for finding the maximum with priced information}\label{apx:maximum_lb}

The instance from~\cite{gupta2001sorting} is the following:
there are two red nodes 1 and 2, and $n-2$ blue nodes.
There are $n-2$ cost 0 edges that show 1 is greater than all the blue nodes.
There are cost 1 edges between~2 and all blue nodes.
There is a cost~$n$ edge between 1 and 2.

The instance distribution is the following:
\begin{itemize}
    \item with probability 2/n: 2 is the maximum (then 1, then all blue)
    \item with probability 1/n: 1 is maximum, then one of the blues ($n-2$ cases), then 2, then all other blues
\end{itemize}
In the first case, the proof has cost~$n$, in each of the second cases, the proof has cost~$1$, it consist of the query that the special blue is greater than 2.

Any deterministic algorithm has the following costs:
If it queries the edge between 1 and 2, it has cost~$n$.
If it queries the edges between 2 and a  blue in any order, the expected number of queries to find the special blue is~$n/2$.

The expected competitive ratio can be calculated as follows.
If the deterministic algorithm queries the 1 - 2 edge, the ratio is $n/n=1$ in the first case, and $n/1=n$ in all other cases, giving an expected ratio of~$2/n+ (n-2)n/n > n-2$.
If the algorithm does not query the 1-2 edge (and relies on the promise of the distribution that if all blue are smaller than 2, then 2 is the maximum), then 
the ratios are $1, 2, 3,\ldots, n-1, n-1$, depending on which of the edges to the blues is the special one.
This gives an expected ratio of roughly $n/2$. This implies that the competitive ratio of any deterministic algorithm on the distribution above is always $\Omega(n)$. The same holds for any randomized algorithm too, by Yao's principle.

\section{Neighborhood-Based Framework and Existing Works}\label{app:framework}

We start this section with the disclaimer that while this neighborhood-based approach of looking at instance-optimality may not have appeared in print, it may nevertheless be known implicitly to researchers in the field of instance optimality. We do not claim  novelty in this section, but felt it worthwhile to show that existing works fit this paradigm.

\noindent\textbf{Adaptive Sorting and Pre-sortedness.} The classical approach to define how a sorting algorithm is adaptive to pre-sortedness of the input~$\I$, is to define a measure, that assigns a number~$\M(\I)$ to it and by this characterizes how difficult or unusual an input is \cite{AdaptiveCastroWood92}.
Inspired by practical adaptive sorting algorithms that work well in practice, different measure were introduced, and one point of the survey~\cite{AdaptiveCastroWood92} is, that ideally $\M$ reflects the practical setting in which the sorting algorithm operates. 
There are some adaptive sorting algorithms that are asymptotically optimally adaptive to several measures, but practically one is paying an overhead for this universality.

All this works nicely fits into our framework of instance optimality. 
For an instance $\I$, we define the neighborhood $\N(\I) = \{ \I' \mid \M(\I') \le \M(\I) \}$, i.e. all instances that are at least as presorted as the instance at hand.
This leads to a set up, where $O(1)$ instance optimality is the same as $O(1)$ adaptiveness, as defined in \cite{AdaptiveCastroWood92}.

Let us explore if the dependency on the measure of pre-sortedness is necessary. 
We claim that it is very important, and that there must be an aspect of ``human judgement'' in what is considered an interesting such measure, and that it is impossible to be adaptive to all measures simultaneously.
To show this, fix an input size~$N$, and any of the natural measures~$\M$ of pre-sortedness, like counting the number of inversions on input of length~$N$.
Now consider for each of the $N!$ permutations~$\pi$ the measure $\M_\pi$, defined as first applying the permutation $\pi$ to the input and then measuring the pre-sortedness of the resulting string. 
Let~$\A$ be an optimal adaptive algorithm for~$\M$ and let $\A_\pi$ be this algorithm, where the input is first permuted according to~$\pi$. 
Clearly, $\A_\pi$ is optimal adaptive for $\M_\pi$. 
In contrast, it is impossible to have an algorithm~$\A$ that is $O(1)$ optimally adaptive to all $\M_\pi$ simultaneously:
For any input $\I$, there exists the permutation $\pi_\I$ that sorts $\I$, such that $\M\pi(\I) = 0$, and hence the running time of~$\A$ must be linear. 
Hence, the running time of~$\A$ must be linear on all inputs, contradicting the well known $n \log n$ lower bound.

\noindent\textbf{Set intersections, unions and differences} Several other examples of instance optimality are given in \cite{demaine2000adaptive}. 
Let us zoom in on binary merging as an instructive example from that paper. 
In merging, we assume that we have two sequences $A$ and $B$ that are already sorted, and the task is to merge, i.e., identify the sorted sequence of the union of the two.
Classically, in the worst case, this can be done with $|A|+|B|-1$ comparisons, by iteratively removing the smallest element.
Observe that in the case where all of $A$ is smaller than all of $B$, a single comparison can be enough.
An \emph{argument} as defined by \cite{demaine2000adaptive} is a set of comparisons that are sufficient to identify the output. 
We would call this a certificate, and it corresponds to the edges on the transitive reduction that are not known a priori.
Intuitively, there is no chance that a general algorithm can get away with a number of comparisons linear in certificate size,
and the classical improved merging algorithm (specially advocated for merging two sequences that are quite different in length) advances by performing exponential searches, not only identifying the current smallest element, but the whole prefix of one sequence that is smaller than the smallest element of the other sequence. Call such a sequence a \emph{run}.
Now the cost of the algorithm is basically the sum of the logarithms of the length of all but the largest runs.
Here \cite{demaine2000adaptive} continue by using the information theory of encoding natural numbers to argue that the comparisons performed bythe merging algorithm are necessary.

In our framework, we would define the neighborhood of an instance with runs of length $r_1,b_1,\ldots$ to be all instances with the same number of runs and only the longest run length potentially increased and all other run length potentially decreased.
Using the fact that comparison based algorithms are branching only in 2 ways, the log of the size of the neighborhood is a lower bound, and this is precisely what \cite{demaine2000adaptive} uses to show optimality of the mentioned algorithm.
Observe that by the discussion above about binary search, we could change the assumptions about what run lengths are most likely, and would get different notions of instance optimality, all of which allow for optimal algorithms, but are different.

\noindent\textbf{Universal optimality} The notion of instance optimality is also the focus of \cite{haeupler2023universal}, where they study Dijkstra's shortest path algorithm and introduce the notion of \emph{Universal Optimality}.
In our framework, we define the neighborhood of a weighted graph~$G$ to be all weighted graphs that share the same underlying graph and only differ in the weight function. 
Now our notion of instance optimality coincides with the notion of universal optimality, both compare the number of comparisons of an algorithm (in \cite{haeupler2023universal} Dijkstra's algorithm with a novel priority queue that implements ``the right'' working set property) with the number of comparisons any algorithm that is correct on the neighborhood needs to perform (worst case on the neighborhood).

Finally, we observe that the work on learning-augmented algorithm \cite{mitzenmacher2020algorithms} may also naturally fit into the neighborhood-based framework, as predictions help the algorithm narrow down the neighborhood that the instance lies in.  %

\end{document}